\documentclass{article}
 \pdfoutput=1
\usepackage[utf8]{inputenc}
\usepackage[english]{babel}
\usepackage[T1]{fontenc}
\usepackage{graphicx}
\usepackage{subfigure}
\usepackage{mathrsfs}
\usepackage{algorithm}
\usepackage{algorithmic}
\usepackage{amsfonts,amsthm,amsmath,amssymb}      
\usepackage[usenames,dvipsnames]{color}
\usepackage{stmaryrd}
\usepackage[toc,page]{appendix}
\usepackage{hyperref}

\theoremstyle{definition}
\newtheorem{definition}{Definition}
\newtheorem{proposition}{Proposition}

\newcounter{para}
\renewcommand{\thepara}{\roman{para}}
\newcommand\myparagraph{\ \newline\par\refstepcounter{para}\textbf{(\thepara)}\space\textbf}

\begin{document}
\title{Graph edit distance : a new binary linear programming formulation}
\author{Julien Lerouge, Zeina~Abu-Aisheh, Romain Raveaux\\Pierre H\'eroux, and S\'ebastien Adam}
\maketitle

\begin{abstract}

Graph edit distance (GED) is a powerful and flexible graph matching paradigm that can be used to address different tasks in structural pattern recognition, machine learning, and data mining. In this paper, some new binary linear programming formulations for computing the exact GED between two graphs are proposed. A major strength of the formulations lies in their genericity since the GED can be computed between directed or undirected fully attributed graphs (i.e. with attributes on both vertices and edges). Moreover, a relaxation of the domain constraints in the formulations provides efficient lower bound approximations of the GED. A complete experimental study comparing the proposed formulations with 4 state-of-the-art algorithms for exact and approximate graph edit distances is provided. By considering both the quality of the proposed solution and the efficiency of the algorithms as performance criteria, the results show that none of the compared methods dominates the others in the Pareto sense. As a consequence, faced to a given real-world problem, a trade-off between quality and efficiency has to be chosen w.r.t. the application constraints. In this context, this paper provides a guide that can be used to choose the appropriate method.

\end{abstract}




\section{Introduction}
\label{sec:intro}

Graphs are data structures able to describe complex entities through their elementary components (the vertices of the graph) and the relational properties between them (the edges of the graph). For attributed graphs, both vertices and edges can be characterized by attributes that can vary from nominal labels to more complex descriptions such as strings or feature vectors, leading to very powerful representations. As a consequence of their inherent genericity and their ability to represent objects as composition of elementary entities, and thanks to the general improvement of computing power, graph representations have become more and more popular in many application domains such as computer vision, image understanding, biology, chemistry, text processing or pattern recognition. With this emergence of the use of graphs, new algorithmic issues have arised such as graph mining \cite{graphMining_Zou}, graph clustering \cite{graphClusteringChen_Djidjev} or graph classification \cite{graphClassification_Wu}. 

A major issue related to the graph-based algorithms mentioned above is the computation of a (dis)similarity measure between two graphs. A huge number of algorithms have been proposed in the literature to solve this problem, which is particularly crucial for machine learning issues. They can be categorized as \textit{embedding-based methods} vs. \textit{matching-based methods}. 

In \textit{embedding-based methods}, the key-idea is to project the input graphs to be compared into a vector space  in order  to benefit from the distance computation  designed  for  vectorial  representations. Among existing approaches, some are based on an implicit projection, through the use of graph kernels \cite{Kernel_Foggia,Kernel_Raviv} whereas other methods make the projection explicit, through the computation of a feature vector for each graph to be compared. The features can result for example from frequencies of appearance of specific sub-structures \cite{ExplicitEmbedding_Kramer,ExplicitEmbedding_Sidere} or from a spectral analysis of the graphs \cite{SpectralEmbedding_Ren,SpectralEmbedding_Shi}. Embedding-based methods are generally computationally effective since they do not involve a complete matching process. On the other hand, they do not take into account the complete relational properties and do not provide the matching between vertices and edges.  

A second way to compute the dissimilarity between two graphs consist in using \textit{matching-based methods}. In such a case, computing the similarity between two graphs requires the computation and the quanfitication of the "best" matching between them. Different kinds of matching algorithm have been used for such a computation. They differ according to the kind of constraints that must be respected and to those that can be relaxed. As an example, maximum common subgraph and/or minimum common supergraphs have been used in \cite{MCS_Bunke,MCS_Fernandez} to derive a graph distance metric. Since exact isomorphism rarely occur in pattern analysis applications, another interesting class of matching problem for similarity evaluation is the error tolerant graph matching problem. A graph matching is said to be error-tolerant when the matching tolerates differences on the topology and/or the  attributes of the vertices and the edges. Adjacency matrix eigendecomposition \cite{Matching_Umeyama} or graduated assignment methods \cite{Matching_Gold,Matching_Wyk}  are examples of methods that  have been used to tackle this problem. Another well known error-tolerant \textit{matching-based method} that can be used to compute a dissimilarity measure between two graphs is the graph edit distance (GED). In this method, the graph matching process and the dissimilarity computation are linked through the introduction of a set of graph edit operations (e.g. node insertion, node deletion). Each edit operation is characterized by a cost, and the graph edit distance is the total cost of the least expensive sequence of edit operations that transforms one graph into the other one. A major advantage of graph edit distance is that it is a dissimilarity measure for arbitrarily structured and arbitrarily  attributed  graphs. In contrast with other approaches, it does not suffer from any restrictions and can be applied to any type of graph, including hypergraphs \cite{Bunke:1983:IGM:2305869.2306079}. Graph edit distance has been used to address various graph classification problems \cite{DBLP:journals/cviu/RaveauxAHT11,Riesen:2007:GEV:1769371.1769413,Bunke2012811}. However, a main drawback of graph edit distance is its computational complexity which is exponential in the number of nodes of the involved graphs. Consequently, computation of graph edit distance is feasible for graphs of rather small size only. In order to overcome this restriction, some number of fast but suboptimal methods have been proposed in the literature (e.g. \cite{DBLP:conf/gbrpr/RiesenNB07, DBLP:journals/ijprai/FankhauserRBD12, DBLP:journals/ivc/RiesenB09, DBLP:conf/sspr/NeuhausRB06, 10.1109/34.862201, DBLP:journals/prl/RaveauxBO10}). On the other hand, only few optimal methods have been proposed to postpone the graph size restriction \cite{DBLP:conf/mlg/RiesenFB07, DBLP:conf/gbrpr/FankhauserRB11}.
 
This paper tackles the problem of graph edit distance computation by proposing two main contributions. The first one consists in solving the graph edit distance problem with binary linear programming. More precisely, two original exact formulations of the GED are provided.
They are very general, since they are able to compute the GED between directed or undirected fully attributed graphs (i.e. with attributes on both vertices and edges). Furthermore, a relaxation of the domain constraints in the formulations provides efficient lower bound approximations of the GED. On the basis of these formulations, the second contribution is a very complete comparative study where eight algorithms for exact and approximate graph edit distances are compared on a set of graph datasets. By considering both the quality of the proposed solution and the efficiency of the algorithms, we show that none of the compared methods dominates the others in the Pareto sense. As a consequence, faced to a given real-world problem, a trade-off between quality and efficiency has to be chosen w.r.t. the application constraints. In this context, this paper provides a guide that can be used to choose the appropriate method.

This paper is organized as follows: Section 2 presents the important definitions necessary for introducing our formulations of the GED. Then, section 3 reviews existing approaches for computing GED with exact and inexact methods. Section 4 describes the proposed binary linear programming formulations. Section 5 presents the experiments and analyses the obtained results. Section 6 provides some concluding remarks.

\section{Problem statement}
\label{sec:pbstatement}


In this paper, we are interested in computing the graph edit distance between attributed graphs. 


\begin{definition}
An attributed graph $G$ is a 4-tuple $G=(V, E, \mu, \xi)$, where :
\begin{itemize}
 \item $V$ is a set of vertices,
 \item $E$ is a set of edges, such that $\forall e=(i,j)  \in E, i \in V \textrm{ and } j \in V$,
 \item $\mu : V \rightarrow L_V$ is a vertex labeling function which associates the label $\mu(v)$ to all vertices $v$ of $V$,
  where $L_V$ is the set of possible labels for the vertices,
 \item $\xi : E \rightarrow L_E$ is an edge labeling function which associates the label $\xi(e)$ to all edges $e$ of $E$,
  where $L_E$ is the set of possible labels for the edges.
\end{itemize}
\label{def:attributedgraph}
\end{definition}

The vertices (resp. edges) label space $L_V$ (resp. $L_E$) may be composed of any combination of numeric, symbolic or string attributes.

A graph $G$ is said simple if it has no loop (an edge that connects a vertex to itself) and no multiedge (several edges between the same vertices). In this case, $E \subseteq \{(i,j) \in V \times V / i \neq j \}$ and an edge can be unambiguously designated by the pair of edges it connects. Otherwise, $G$ is a multigraph and $E$ is a multiset.
A graph $G$ is said undirected if the relation E is symmetric, i.e. if its edges have no orientation. In this case, $\forall (i, j) \in E, (j, i) \in E \textrm{ and } (i,j) = (j,i)$. Otherwise, $G$ is a directed graph.



Definition \ref{def:attributedgraph} allows us to handle arbitrarily structured graphs (directed or undirected, simple graphs or multigraphs) with unconstrained labeling functions.

Many applications using graph-based representations need to evaluate how two graphs are similar, or how they differ. The graph edit distance is commonly used to measure the dissimilarity between two graphs. Graph edit distance is an error-tolerant graph matching method. It defines the dissimilarity of two graphs by the minimum amount of distortion that is needed to transform one graph into another  \cite{Bunke:1983:IGM:2305869.2306079}.

\begin{definition}
  The graph edit distance $d(.,.)$ is a function 
  \begin{eqnarray*}
    d & : & \mathcal{G} \times \mathcal{G} \rightarrow \mathbb{R}^+ \\
    & & (G_1,G_2) \mapsto d(G_1,G_2) =\\
    & & \min_{o=(o_1,\ldots,o_k)\in \Gamma(G_1,G_2)} \sum_{i=1}^kc(o_i) 
  \end{eqnarray*}
\end{definition}
\noindent where $G_1=(V_1,E_1, \mu_1, \xi_1)$ and $G_2=(V_2,E_2, \mu_2, \xi_2)$ are two graphs from the set $\mathcal{G}$ and $\Gamma(G_1,G_2)$ is the set of all edit paths $o=(o_1,\ldots,o_k)$ allowing to transform $G_1$ into $G_2$. An elementary edit operation $o_i$ is one of vertex substitution ($v_1 \rightarrow v_2$), edge substitution ($e_1 \rightarrow e_2$), vertex deletion ($v_1 \rightarrow \epsilon$), edge deletion: ($e_1 \rightarrow \epsilon$), vertex insertion ($ \epsilon \rightarrow v_2$) and edge insertion ($ \epsilon \rightarrow e_2$) with $v_1 \in V_1$, $v_2 \in V_2$, $e_1 \in E_1$ and $e_2 \in E_2$. $\epsilon$ is a dummy vertex or edge which is used to model insertion or deletion. $c(.)$ is a cost function on elementary edit operations $o_i$ that satisfies 
\begin{itemize}
\item $c(v_1 \rightarrow v_2) \leq c(v_1 \rightarrow v) + c(v \rightarrow v_2)$
\item $c(e_1 \rightarrow e_2) \leq c(e_1 \rightarrow e) + c(e \rightarrow e_2)$
\item $c(v_1 \rightarrow \epsilon) \leq c(v_1 \rightarrow v) + c(v \rightarrow \epsilon)$
\item $c(e_1 \rightarrow \epsilon) \leq c(e_1 \rightarrow e) + c(e \rightarrow \epsilon)$
\item $c(\epsilon \rightarrow v_2) \leq c(\epsilon \rightarrow v) + c(v \rightarrow v_2)$
\item $c(\epsilon \rightarrow e_2) \leq c(\epsilon \rightarrow e) + c(e \rightarrow e_2)$
\end{itemize}

Moreover, in order to guarantee the symmetry property ($d(G_1,G_2)=d(G_2,G_1)$), the reverse edit path should result in the same cost. So, these costs have to be defined in a symmetric manner so that $c(v_1 \rightarrow v_2) = c(v_2 \rightarrow v_1)$, $c(e_1 \rightarrow e_2) = c(e_2 \rightarrow e_1) $, $c(v \rightarrow \epsilon) = c(\epsilon \rightarrow v)$ and $c(e \rightarrow \epsilon) = c(\epsilon \rightarrow e)$.


When the graph edit distance is computed between unlabeled graphs, the identity property ($d(G_1,G_2)=0 \Leftrightarrow G_1=G_2$) imposes that the substitution costs are equals to 0. The insertion and deletion costs are then set to a constant. 
In the more general case where the graph edit distance is computed between attributed graphs, edit costs are generally defined as functions of vertices (resp. edges) labels. More precisely, substitution costs are defined as a function of the labels of the substituted vertices (resp. edges), whereas insertion and deletion are penalized with a value linked to the label of the inserted/deleted vertex (resp. edge).

\begin{equation*}
  \begin{array}{l}
    c(v_1 \rightarrow v_2) = c(v_2 \rightarrow v_1) = f_v(\mu_1(v_1),\mu_2(v_2)) \\
    c(e_1 \rightarrow e_2) = c(e_2 \rightarrow e_1) = f_e(\xi_1(e_1),\xi_2(e_2)) \\
    c(v \rightarrow \epsilon) = c(\epsilon \rightarrow v) = g_v(\mu(v)) \\
    c(e \rightarrow \epsilon) = c(\epsilon \rightarrow e) = g_e(\xi(e))
  \end{array}
\end{equation*}

\section{Related work}
\label{sec:related}
The graph edit distance, which is the minimum cost associated to an error correcting graph matching, has been the subject of many studies in the literature. Several papers propose surveys of these works \cite{citeulike:809181,livi13,Gao:2010:SGE:1714377.1714380}. They distinguish exact approaches from approximations. Indeed, as stated in \cite{Zeng09comparingstars:}, the graph edit distance problem is NP-hard. It is then prohibitively difficult to compute the graph edit distance for large graphs, and the literature reports exact methods to compute GED only for small graphs, while approximations by means of upper and lower bounds computation are often used for larger graphs.

\subsection{Exact approaches}

A first family of exact computation of the graph edit distance is based on the widely known A$^*$ algorithm. This algorithm relies on the exploration of the tree of solutions. In this tree, each node corresponds to a partial edition of the graph. A leaf of the tree corresponds to an edit path which transforms one of the input graphs into the other one. The exploration of the tree is guided by developing most promising ways on the basis of an estimation of the graph edit distance. For each node, this estimation is the sum of the cost associated to the partial edit path and an estimation of the cost for the remaining path, the latter being given by a heuristic. Provided that the estimation of the future cost is lower than or equal to the real cost, an optimal path from the root node to a leaf node is guaranteed to be found \cite{4082128}. A simple way to fulfill this constraint would be to set the estimation of the future cost to zero, but this may lead to explore the whole tree of solutions. Indeed, the smaller the difference between the estimation and the real future cost, the fewer nodes will be expanded by the A* algorithm. However, the other extreme which consists in computing the real cost for the remaining edit path would require an exponential time. The different A*-based methods published in the literature mainly differ in the implemented heuristics for the furture cost estimation which correspond to different tradeoffs between approximation quality and their computation time \cite{DBLP:conf/mlg/RiesenFB07,Fischer2015331}.

In an other family of algorithms, the graph edit distance is computed by solving a binary linear program. Almohamad and Duffuaa \cite{Almohamad:1993:LPA:628301.628477} propose a binary linear programming formulation of the weighted graph matching problem which aims at determining the permutation matrix minimizing the $L_1$ norm of the difference between adjacency matrix of the input graph and the permuted adjacency matrix of the target one. Later, Justice and Hero \cite{Justice:2006:BLP:1155317.1155424} also proposed a BLP formulation of the graph edit distance problem aiming at determining the permutation matrix which minimizes the cost of transforming $G_1$ into $G_2$, with $G_1$ and $G_2$ two unweighted and undirected graphs. The criterion to be minimized (see eq. \ref{eq:blp}) takes into account costs for matching vertices, but the formulation does not integrate the ability to process graphs that carry labels on their edges. 
\begin{equation}
  d(G_1,G_2)=\min_P\sum_{i=1}^n\sum_{j=1}^nC_{i,j}P_{i,j} + \frac{1}{2}\left\|A_1-PA_2P^T\right\|_1\label{eq:blp}
\end{equation}
\noindent where $C_{i,j}$ is the cost for matching the $i^{th}$ vertex in $G_1$ and the $j^{th}$ vertex in $G_2$. $A_1$ (resp. $A_2$) is the adjacency matrix of $G_1$ (resp. $G_2$), and $P$ is an orthogonal permutation matrix such that $PP^T=P^TP=I$. A mathematical transformation is used to transform this non linear optimization problem into a linear one. The modeling of graphs by means of adjacency matrix restricts the formulation to the processing of simple graphs.

\subsection{Approximations}

Considering that exact computation of graph edit distance can be performed in a reasonable time only for small graphs, many researchers have focused their effort on the computation of approximations in polynomial time. For example, in their paper \cite{Justice:2006:BLP:1155317.1155424}, Justice and Hero have proposed a lower bound of the graph edit distance which can be computed in $\mathcal{O}(n^7)$ by extending the domain of variables in $P$ from $\{0,1\}$ to $[0,1]$. In the same paper, they also proposed an upper bound that can be computed in $\mathcal{O}(n^3)$ by determining vertex correspondance based only on the vertex term of eq. \ref{eq:blp} thanks to the Hungarian method (also called Munkres assignment algorithm). The remaining part of the cost is deduced from the permutation matrix determined in the previous step. In a quite similar way, Riesen \emph{et al.} \cite{DBLP:journals/ivc/RiesenB09} propose to first exploit a cost matrix for vertex substitution, insertion or deletion in order to determine the vertex assignment thanks to the Munkres algorithm with a complexity of $\mathcal{O}((n_1+n_2)^3)$ in the number of nodes $n_1=|V_1|$ and $n_2=|V_2|$ of the involved graphs. The vertex assignment is then used to infer an edit path which transforms one graph into the other and whose associated cost is an upper bound of the graph edit distance.

In \cite{DBLP:conf/sspr/NeuhausRB06}, Neuhaus \emph{et al.} propose another approximation based on $A^*$-based method. The first one, called $A^*$-\textsc{beamsearch}, propose to prune the tree of solutions by limiting the number of concurrent partial solutions to the $q$ most promising ones. At the end of the algorithm, a valid edit path and its associated cost are provided, but there is no guarantee that it corresponds to the optimal one, since the latter may have been eliminated in earlier steps of the algorithm. The parameter $q$, corresponding to the number of concurrent partial solutions to keep, allows to manage the trade-off between combinatorial cost and quality of the approximation. This method provides an upper bound of the exact graph edit distance. In the same paper, a method called $A^*$-\textsc{pathlength} proposes to speed up the access to a leaf node in the tree of solutions by giving a higher exploration priority to long partial edit paths. This strategy is motivated by the observation that first assignments are the most computationally expensive and that they are rarely called into question.

More recently, in \cite{DBLP:dblp_conf/sspr/RiesenFB14}, the vertex assignment computed by means of bipartite graph matching is used as an initialization step for a genetic algorithm which attempts to improve the quality of the approximation. Indeed, from any vertex assignment, it is possible to derive an edit path and finally compute its cost \cite{DBLP:journals/ivc/RiesenB09}. The vertex assignment which is optimal in terms of vertex subtitution is not always optimal for the whole edit path. However, it has been observed that it may only differ with few assignments. In the proposed genetic algorithms, population individuals correspond to different vertex mappings. The initial population is generated by deriving mappings that are mutated version of the one that has been determined by the hungarian algorithm. The probability of a vertex mapping to be selected is linked to the vertex substitution cost. The lower the corresponding edit distance, the best the individual fits the objective function. The genetic algorithm iterates by selecting and mixing several mappings.

Fischer \emph{et al.} \cite{Fischer2015331} propose to integrate in the A* algorithm a heuristic based on a modifed Hausdorff distance. Given two graphs $G_1$ and $G_2$ and $C$ a cost matrix for vertex substitution\footnote{It also integrates vertex insertion and deletion costs}, the Hausdorff Edit Distance is defined by 

$$
HED(G_1,G_2,C) = \sum_{u \in V_1} \min_{v \in V_1 \cup \epsilon} C(u,v) +  \sum_{v \in V_2} \min_{u \in V_1 \cup \epsilon} C(v,u)
$$
\noindent which can be interpreted as the sum of distances to the most similar vertex in the other graph. This distance is computed in a time complexity of $\mathcal{O}(n_1.n_2)$. 

Graph edit distance approximations have also been proposed in a probabilistic framework \cite{10.1109/34.862201,Wilson:1997:SMD:262631.262639} where the objective is to find the vertex assignment that maximizes the a posteriori probability considering vertex attributes. However, unlike the methods formerly presented, the corresponding heuristics are unbounded and can not be exploited by branch and bound algorithms to prune the tree of solutions or to efficiently prioritize its exploration in the A* algorithm.

\section{Graph edit distance using binary linear programming}
\label{sec:proposal}
In this article, the graph edit distance problem is modeled by a Binary Linear Program (BLP). A BLP is a restriction of integer linear programming (ILP) where the variables are binary. Hence, its general form is as follows:
\begin{subequations}
\begin{align}
\min_{x}&~{c^T x} \label{ilp:o}\\
\text{subject to }&
Ax \leq b \label{ilp:c1}\\
&
 x \in \{0,1\}^n \label{ilp:c2}
\end{align}
\end{subequations}
where $c \in \mathbb{R}^{n}, A \in \mathbb{R}^{n \times m} \text{ and } b \in \mathbb{R}^{m}$ are data of the problem.
A solution of this optimization problem is a vector $x$ of $n$ binary variables. $A$ is used to express linear inequality constraints \eqref{ilp:c1}.
If the program is feasible, i.e. if it has such solutions, then the optimal solution is the one that minimizes the objective function \eqref{ilp:o} and
respects constraints \eqref{ilp:c1} and \eqref{ilp:c2}. The objective function $c^T x$ is a linear combination of variables of $x$ weighted by the components of the vector $c$. 

In this section, we present the two formulations we wrote for the GED. Then, we present how the formulations are solved. Finally, we discuss how the relaxation of the formulations can provide a lower bound of the GED.  

\subsection{Modelling the GED problem}

In this subsection, we first define in \ref{subsubsection:variables} the variables used for formulating the GED as a BLP. Then, we describe in  \ref{subsubsection:objective} the objective function of the program and in \ref{subsubsection:constraints} the linear constraints that must be satisfied to correctly match the two graphs.

\subsubsection{Variable and cost functions definitions}
\label{subsubsection:variables}

Our goal is to compute the graph edit distance between two graphs $G_{1} = (V_1,E_1,\mu_1,\xi_1) $ and $G_{2} = (V_2,E_2,\mu_2,\xi_2)$. In the rest of this section, for the sake of simplicity of notations, we consider that the graphs $G_{1}$ and $G_{2}$ are simple directed graphs. However, let us emphasize that the formulations given in this section can be applied without modification to multigraphs, and that the undirected case only needs some slight modifications (please refer to appendix \ref{appendix:undirected}).

In the GED definition provided in section \ref{sec:pbstatement}, the edit operations that are allowed to match the graphs $G_{1}$ and $G_{2}$ are (i) the substitution of the label of a vertex (respectively an edge) of $G_{1}$ with the label of a vertex (resp. an edge) of $G_{2}$, (ii) the deletion of a vertex (or an edge) from $G_{1}$ and (iii) the insertion of a vertex (or an edge) of $G_{2}$ in $G_{1}$. For each type of edit operation, we define a set of corresponding binary variables:

\begin{itemize}
 \item $\forall (i,k) \in V_1 \times V_2,\\ 
   x_{i,k} = \left\{\begin{array}{l}
         1 \text{ if $i$ is substituted with $k$,}\\
         0 \text{ otherwise.}\\ \end{array} \right.$
 \item $\forall (ij,kl) \in E_1 \times E_2, \\
   y_{ij,kl} = \left\{\begin{array}{l}
         1 \text{ if $ij$ is substituted with $kl$,}\\
         0 \text{ otherwise.}\\ \end{array} \right.$
 \item $\forall i \in V_1, u_i = \left\{\begin{array}{l}
         1 \text{ if $i$ is deleted from } G_{1}\\
         0 \text{ otherwise.}\\ \end{array} \right.$
 \item $\forall ij \in E_1, e_{ij} = \left\{\begin{array}{l}
         1 \text{ if $ij$ is deleted from } G_{1}\\
         0 \text{ otherwise.}\\ \end{array} \right.$
 \item $\forall k \in V_2, v_k = \left\{\begin{array}{l}
         1 \text{ if $k$ is inserted in } G_{1}\\
         0 \text{ otherwise.}\\ \end{array} \right.$
 \item $\forall kl \in E_2, f_{kl} = \left\{\begin{array}{l}
         1 \text{ if $kl$ is inserted in } G_{1}\\
         0 \text{ otherwise.}\\ \end{array} \right.$
\end{itemize}

Using these notations, we define an edit path between $G_{1}$ and $G_{2}$ as a 6-tuple $(\mathbf{x}, \mathbf{y}, \mathbf{u}, \mathbf{v}, \mathbf{e}, \mathbf{f})$ where $\mathbf{x} = (x_{i,k})_{(i,k) \in V_1 \times V_2}$, $\mathbf{y} = (y_{ij,kl})_{(ij,kl) \in E_1 \times E_2}$, $\mathbf{u} = (u_i)_{i \in V_1}$, $\mathbf{e} = (e_{ij})_{ij \in E_1}$, $ \mathbf{v} = (v_k)_{k \in V_2}$ and $\mathbf{f} = (f_{kl})_{kl \in E_2}$.


In order to evaluate the global cost of an edit path, elementary costs for each edit operation must be defined. We adopt the following notations for these costs:

\begin{itemize}
 \item $\forall (i,k) \in V_1 \times V_2, c(i \rightarrow k)$ is the cost of substituting the vertex $i$ with $k$,
 \item $\forall (ij,kl) \in E_1 \times E_2, c(ij \rightarrow kl)$ is the cost of substituting the edge $ij$ with $kl$,
 \item $\forall i \in V_1, c(i \rightarrow \epsilon)$ is the cost of deleting the vertex $i$ from $G_{1}$,
 \item $\forall ij \in E_1, c(ij \rightarrow \epsilon)$ is the cost of deleting the edge $ij$ from $G_{1}$,
 \item $\forall k \in V_2, c(\epsilon \rightarrow k)$ is the cost of inserting the vertex $k$ in $G_{1}$,
 \item $\forall kl \in E_2, c(\epsilon \rightarrow kl)$ is the cost of inserting the edge $kl$ in $G_{1}$.
\end{itemize}

These cost functions traditionally depend on the labels of the vertices and of the edges. 
Table \ref{tab:synthnotations} gives a summary of the notations.

\begin{table}[!h]
\centering
\begin{tabular}{|c|c|c|c|c|c|}
\hline
Type & Edit operation & $G_1$ & $G_2$ & Cost & Variable \\ \hline
Vertex & Substitution & $i$ & $k$ & $c(i \rightarrow k)$ & $x_{i,k}$ \\ \hline
Vertex & Deletion & $i$ & $\times$ & $c(i \rightarrow \epsilon)$ & $u_i$ \\ \hline
Vertex & Insertion & $\times$ & $k$ & $c(\epsilon \rightarrow k)$ & $v_k$ \\ \hline
Edge & Substitution & $ij$ & $kl$ & $c(ij \rightarrow kl)$ & $y_{ij,kl}$ \\ \hline
Edge & Deletion & $ij$ & $\times$ & $c(ij \rightarrow \epsilon)$ & $e_{ij}$ \\ \hline
Edge & Insertion & $\times$ & $kl$ & $c(\epsilon \rightarrow kl)$ & $f_{kl}$ \\ \hline
\end{tabular}
\caption{\label{tab:synthnotations} Summary of the notations for the GED framework}
\end{table}

\subsubsection{Objective function}
\label{subsubsection:objective}
The objective function \eqref{eq:objective} is the overall cost induced by applying an edit path $(\mathbf{x}, \mathbf{y}, \mathbf{u},
\mathbf{v}, \mathbf{e}, \mathbf{f})$ that transforms a graph $G_{1}1$ into a graph $G_{2}$, using the elementary costs of table \ref{tab:synthnotations}. In order to get
the graph edit distance between $G_{1}$ and $G_{2}$, this cost must be minimized.

\begin{equation}
\begin{aligned}
  \min_{\mathbf{x,y,u,v,e,f}} \Biggl( &\sum_{i \in V_1}\sum_{k \in V_2} c(i \rightarrow k) \cdot x_{i,k}\\& + \sum_{ij \in E_1}\sum_{kl \in E_2} c(ij \rightarrow kl) \cdot y_{ij,kl} \\
				      &+ \sum_{i \in V_1} c(i \rightarrow \epsilon) \cdot u_i + \sum_{k \in V_2} c(\epsilon \rightarrow k) \cdot v_k \\
				      &+ \sum_{ij \in E_1} c(ij \rightarrow \epsilon) \cdot e_{ij} + \sum_{kl \in E_2} c(\epsilon \rightarrow kl) \cdot f_{kl} \Biggr)
\end{aligned}
\label{eq:objective}
\end{equation}

\subsubsection{Constraints}
\label{subsubsection:constraints}
The constraints presented in this part are designed to guarantee that the admissible solutions of the BLP
are edit paths that transform $G_{1}$ in a graph which is isomorphic to $G_{2}$. An edit
path is considered as admissible if and only if the following conditions are respected:

\begin{enumerate}
 \item it provides a one-to-one mapping between a subset of the vertices of $G_{1}$ and a subset
	    of the vertices of $G_{2}$. The remaining vertices are either deleted or inserted,
 \item it provides a one-to-one mapping between a subset of the edges of $G_{1}$ and a subset of
	    the edges of $G_{2}$. The remaining edges are either deleted or inserted,
 \item the vertices matchings and the edges matchings are consistent, i.e. the graph
	    topology is respected.
\end{enumerate}
\renewcommand{\theenumi}{\alph{enumi}}

The following paragraphs describes the linear constraints used to integrate these conditions into the BLP.

\myparagraph{Vertices matching constraints}
The constraint \eqref{eq:sumx1} ensures that each vertex of $G_{1}$ is either matched to exactly one
vertex of $G_{2}$ or deleted from $G_{1}$, while the constraint \eqref{eq:sumx2} ensures that each
vertex of $G_{2}$ is either matched to exactly one vertex of $G_{1}$ or inserted in $G_{1}$:

\begin{equation} \label{eq:sumx1}
  u_i + \sum_{k \in V_2} x_{i,k}  = 1 \quad \forall i \in V_1
\end{equation}
\begin{equation} \label{eq:sumx2}
  v_k + \sum_{i \in V_1} x_{i,k}  = 1 \quad \forall k \in V_2
\end{equation}

\myparagraph{Edges matching constraints}
Similarly to the vertex matching constraints, the constraints \eqref{eq:sumy1} and \eqref{eq:sumy2} guarantee
a valid mapping between the edges:
\begin{equation} \label{eq:sumy1}
  e_{ij} + \sum_{kl \in E_2} y_{ij,kl}  = 1 \quad \forall ij \in E_1
\end{equation}
\begin{equation} \label{eq:sumy2}
  f_{kl} + \sum_{ij \in E_1} y_{ij,kl}  = 1 \quad \forall kl \in E_2
\end{equation}

\myparagraph{Topological constraints}
The respect of the graph topology in the matching of the vertices and of the edges is described in the following proposition~:

\begin{proposition}
 An edge $ij \in E_1$ can be matched to an edge $kl \in E_2$ if and only if the head vertices $i \in V_1$ and $k \in V_2$, on the one hand,
 and if the tail vertices $j \in V_1$ and $l \in V_2$, on the other hand, are respectively matched.
\end{proposition}

This quadratic constraint can be expressed linearly with the following constraints \eqref{eq:topology_1} and \eqref{eq:topology_2}:
\begin{itemize}
 \item 	$ij$ and $kl$ can be matched if and only if their head vertices are matched:
	\begin{equation}
	  y_{ij,kl} \leq x_{i,k} \quad \forall (ij, kl) \in E_1 \times E_2
	  \label{eq:topology_1}
	\end{equation}
 \item 	$ij$ and $kl$ can be matched if and only if their tail vertices are matched:
	\begin{equation}
	  y_{ij,kl} \leq x_{j,l} \quad \forall (ij, kl) \in E_1 \times E_2
	  \label{eq:topology_2}
	\end{equation}
\end{itemize}

\subsubsection{Straightforward formulation}
Putting equations \ref{eq:objective} to \ref{eq:topology_2} altogether leads to a first straightforward version of the BLP formulation:
\begin{subequations}
  \begin{center}(F1)\end{center}
  \begin{equation}
  \begin{aligned}
    \min_{\mathbf{x,y,u,v,e,f}} \Biggl( &\sum_{i \in V_1}\sum_{k \in V_2} c(i \rightarrow k) \cdot x_{i,k} \\
    & + \sum_{ij \in E_1}\sum_{kl \in E_2} c(ij \rightarrow kl) \cdot y_{ij,kl} \\
					&+ \sum_{i \in V_1} c(i \rightarrow \epsilon) \cdot u_i + \sum_{k \in V_2} c(\epsilon \rightarrow k) \cdot v_k \\
					&+ \sum_{ij \in E_1} c(ij \rightarrow \epsilon) \cdot e_{ij} + \sum_{kl \in E_2} c(\epsilon \rightarrow kl) \cdot f_{kl} \Biggr)
    \end{aligned}
    \label{f1:o}
  \end{equation}
  \begin{align}
    \text{subject to}\quad
    &u_i + \sum_{k \in V_2} x_{i,k}  = 1 \quad \forall i \in V_1\label{f1:c1}\\
    & v_k + \sum_{i \in V_1} x_{i,k} = 1 \quad \forall k \in V_2\label{f1:c2}\\
    &e_{ij} +\sum_{kl \in E_2} y_{ij,kl}  = 1 \quad \forall ij \in E_1\label{f1:c3}\\
    &f_{kl}+\sum_{ij \in E_1} y_{ij,kl}  = 1 \quad \forall kl \in E_2\label{f1:c4}\\
    &y_{ij,kl} \leq x_{i,k} \quad \forall (ij, kl) \in E_1 \times E_2\label{f1:c5}\\
    &y_{ij,kl} \leq x_{j,l} \quad \forall (ij, kl) \in E_1 \times E_2\label{f1:c6}\\
    \text{with}\quad
    &x_{i,k} \in \{0, 1\} \quad \forall (i, k) \in V_1 \times V_2\label{f1:d1}\\
    &y_{ij,kl} \in \{0, 1\} \quad \forall (ij, kl) \in E_1 \times E_2\label{f1:d2}\\
    &u_i \in \{0, 1\} \quad \forall i \in V_1\label{f1:d3}\\
    &v_k \in \{0, 1\} \quad \forall k \in V_2\label{f1:d4}\\
    &e_{ij} \in \{0, 1\} \quad \forall ij \in E_1\label{f1:d5}\\
    &f_{kl} \in \{0, 1\} \quad \forall kl \in E_2\label{f1:d6}
  \end{align}
\end{subequations}

The domain constraints, from \eqref{f1:d1} to \eqref{f1:d6}, are used to ensure that the solution is binary. 
Thus, the formulation (F1) has:
\begin{itemize}
  \item $|V_1| + |V_2| + |E_1| + |E_2| + |V_1|\cdot|V_2| + |E_1|\cdot|E_2|$ variables,
  \item $|V_1| + |V_2| + |E_1| + |E_2| + 2\cdot|E_1|\cdot|E_2|$ constraints (without the domain constraints).
\end{itemize}

\subsection{Reducing the size of the formulation}

In this subsection, we present a formulation that has been derived from the formulation (F1). We show that this formulation reduces the number of variables and the number of constraints. It will be shown in section \ref{sec:experiments} that this new formulation is more efficient.

\subsubsection{Reducing the number of variables}
\label{subsubsection:reduce_variables}
In the formulation (F1), the variables \textbf{u}, \textbf{v}, \textbf{e} and \textbf{f} help the reader to
understand how the objective function and the contraints were obtained, but they are unnecessary to solve
the GED problem.

We transform the vertex matching constraints \eqref{eq:sumx1} and \eqref{eq:sumx2} into inequality constraints,
without changing their role in the program. As a side effect, it removes the \textbf{u} and \textbf{v} variables
from the constraints:
\begin{equation} \label{eq:sumx1b}
  \sum_{k \in V_2} x_{i,k} \leq 1 \quad \forall i \in V_1
\end{equation}
\begin{equation} \label{eq:sumx2b}
  \sum_{i \in V_1} x_{i,k} \leq 1 \quad \forall k \in V_2
\end{equation}

We do the same for edge matching constraints \eqref{eq:sumy1} and \eqref{eq:sumy2}:
\begin{equation} \label{eq:sumy1b}
  \sum_{kl \in E_2} y_{ij,kl} \leq 1 \quad \forall ij \in E_1
\end{equation}
\begin{equation} \label{eq:sumy2b}
  \sum_{ij \in E_1} y_{ij,kl} \leq 1 \quad \forall kl \in E_2
\end{equation}

We then replace $\mathbf{u}, \mathbf{v}, \mathbf{e} \text{ and } \mathbf{f}$ variables in the objective function
\eqref{eq:objective} by their expressions, which can be easily deduced from equations \eqref{eq:sumx1},
\eqref{eq:sumx2}, \eqref{eq:sumy1} and \eqref{eq:sumy2}:

\begin{equation}
 \begin{aligned}
  &\sum_{i \in V_1}\sum_{k \in V_2} c(i \rightarrow k) \cdot x_{i,k} + \sum_{ij \in E_1}\sum_{kl \in E_2} c(ij \rightarrow kl) \cdot y_{ij,kl} \\& + \sum_{i \in V_1} c(i \rightarrow \epsilon) \cdot u_i + \sum_{k \in V_2} c(\epsilon \rightarrow k) \cdot v_k \\
  & + \sum_{ij \in E_1} c(ij \rightarrow \epsilon) \cdot e_{ij} + \sum_{kl \in E_2} c(\epsilon \rightarrow kl) \cdot f_{kl} \\
  &= \sum_{i \in V_1}\sum_{k \in V_2} (c(i \rightarrow k) - c(i \rightarrow \epsilon) - c(\epsilon \rightarrow k)) \cdot x_{i,k} \\
  &+ \sum_{ij \in E_1}\sum_{kl \in E_2} (c(ij \rightarrow kl) - c(ij \rightarrow \epsilon) - c(\epsilon \rightarrow kl) \cdot y_{ij,kl} + C \\
  &\biggl( \text{with } C = \sum_{i \in V_1} c(i \rightarrow \epsilon) + \sum_{k \in V_2} c(\epsilon \rightarrow k) \\
  & + \sum_{ij \in E_1} c(ij \rightarrow \epsilon) + \sum_{kl \in E_2} c(\epsilon \rightarrow kl)\biggr)
\end{aligned}
\label{eq:objective2}
\end{equation}

As all insertion and deletion variables can be \emph{a posteriori} deduced from the substitution variables, the constraints \eqref{eq:sumx1b} to \eqref{eq:sumy2b}
describe exactly the same set of edit paths than the constraints \eqref{eq:sumx1} to \eqref{eq:sumy2}. Equation \eqref{eq:objective2} shows that the GED can be obtained
without explicitly computing the variables $\mathbf{u}, \mathbf{v}, \mathbf{e} \text{ and } \mathbf{f}$.

\subsubsection{Reducing the number of constraints}
\label{subsubsection:reduce_constraints}
In the formulation (F1), the number of topological constraints, \eqref{eq:topology_1} and \eqref{eq:topology_2},
is $|E_1|\cdot|E_2|$. Therefore, in average, the number of constraints grows quadratically with the mean density
of the graphs. We show that it is possible to formulate the GED problem with potentially less constraints, leaving
the set of solutions unchanged. To this end, we propose to mathematically express Proposition 1 in another way.
We replace the constraints \eqref{eq:topology_1} and \eqref{eq:topology_2} by the following ones:

\begin{itemize}
 \item Given an edge $ij \in E_1$ and a vertex $k \in V_2$, there is at most one edge whose initial vertex is $k$
       that can be matched with $ij$:
       \begin{equation}
	  \sum_{kl \in E_2} y_{ij,kl} \leq x_{i,k} \quad \forall k \in V_2, \forall ij \in E_1\\
	  \label{eq:topology_3}
       \end{equation}
 \item Given an edge $ij \in E_1$ and a vertex $l \in V_2$, there is at most one edge whose terminal vertex is $l$
       that can be matched with $ij$:
       \begin{equation}
	  \sum_{kl \in E_2} y_{ij,kl} \leq x_{j,l} \quad \forall l \in V_2, \forall ij \in E_1\\
	  \label{eq:topology_4}
       \end{equation}
\end{itemize}

\begin{proposition}
Let $\Gamma_1$ be the set of edit paths (between $G_1$ and $G_2$) implied by the set of admissible solutions of (F1),
and let $\Gamma_2$ be the set of edit paths obtained similarly by replacing in (F1) the constraints \eqref{eq:topology_1} and \eqref{eq:topology_2}
by the constraints \eqref{eq:topology_3} and \eqref{eq:topology_4}. Then $\Gamma_1$ = $\Gamma_2$.
\end{proposition}
\begin{proof}\
 \item $\Gamma_2 \subseteq \Gamma_1$:  Let $ij \in E_1$ and $kl \in E_2$, and let us suppose that \eqref{eq:topology_3} is satisfied.
 \begin{align*}
		   &x_{i,k} \geq \sum_{kl' \in E_2} y_{ij,kl'} \\
 \Rightarrow \quad &x_{i,k} \geq y_{ij,kl} + \sum_{kl' \in E_2, kl' \neq kl}  y_{ij,kl'} \\
 \Rightarrow \quad &x_{i,k} \geq y_{ij,kl}
 \end{align*}
 
Thus, the constraint \eqref{eq:topology_1} is satisfied for all $ij \in E_1$ and for all $kl \in E_2$.
Similarly, we deduce that \eqref{eq:topology_2} is satisfied using the constraint \eqref{eq:topology_4}.
 
\item $\Gamma_1 \subseteq \Gamma_2$:  Let $ij \in E_1$ and $k \in V_2$.
 
If $\{l \in V_2: kl \in E_2\} = \emptyset$, then $\sum_{kl \in E_2} y_{ij,kl} = 0$
and \eqref{eq:topology_3} is satisfied. Otherwise, using the constraint \eqref{eq:topology_1}, we have:
\[\forall kl \in E_2, x_{i,k} \geq y_{ij,kl} \; \Rightarrow \; x_{i,k} \geq \max_{kl \in E_2} (y_{ij,kl})\]
  
Constraint \eqref{eq:sumy1} ensures that $\mathrm{card} \{l' \in V_2: y_{ij,kl} = 1\} \leq 1$, thus:
		   
$$\max_{kl' \in E_2} (y_{ij,kl'}) = \sum_{kl' \in E_2} y_{ij,kl'} \Rightarrow \quad x_{i,k} \geq \sum_{kl' \in E_2} y_{ij,kl'} $$
   
\noindent and \eqref{eq:topology_3} is still satisfied.
  
Thus, the constraint \eqref{eq:topology_3} is satisfied for all $ij \in E_1$ and for all $k \in V_2$.
Similarly, we prove that \eqref{eq:topology_4} is satisfied using \eqref{eq:topology_2} and \eqref{eq:sumy2}.
\end{proof}

The number of topological constraints, \eqref{eq:topology_3} and \eqref{eq:topology_4}, is now $|E_1|\cdot|V_2|$.
In average, it grows linearly with the density of the graphs. This leads to substantially shorter formulations of
the GED as the number of graph vertices and edges grows.

Please note that another substitution of constraints \eqref{eq:topology_1} and \eqref{eq:topology_2} is possible,
namely with the two following constraints:

\begin{equation}
  \sum_{ij \in E_1} y_{ij,kl} \leq x_{i,k} \quad \forall i \in V_1, \forall kl \in E_2\label{eq:topology_5}\\
\end{equation}

\begin{equation}
  \sum_{ij \in E_1} y_{ij,kl} \leq x_{j,l} \quad \forall j \in V_1, \forall kl \in E_2\label{eq:topology_6}\\
\end{equation}
This leads to a strictly equivalent formulation in terms of admissible solutions, however it changes the number
of topological constraints, \eqref{eq:topology_5} and \eqref{eq:topology_6}, that would be $|E_2|\cdot|V_1|$.

In addition, we prove that the constraints \eqref{eq:sumy1b} and \eqref{eq:sumy2b} are not necessary to the
formulation of the GED problem, since they are implied by other constraints of the BLP.

\begin{proposition}
Constraint \eqref{eq:sumy1b} is implied by \eqref{eq:sumx1b} and \eqref{eq:topology_3}
\end{proposition}

\begin{proof}\
Let $ij \in E_1$. Given \eqref{eq:topology_3}, we have:
\[ \sum_{kl \in E_2} y_{ij,kl} \leq x_{i,k} \quad \forall k \in V_2 \]
\[ \Rightarrow \sum_{k \in V_2} \sum_{kl \in E_2} y_{ij,kl} \leq \sum_{k \in V_2} x_{i,k} \]

We reduce the left term of this inequation and we use \eqref{eq:sumx1b}:
\[ \sum_{kl \in E_2} y_{ij,kl} \leq \sum_{k \in V_2} x_{i,k} \leq 1 \]

Thus, \eqref{eq:sumy1b} is implied by \eqref{eq:sumx1b} and \eqref{eq:topology_3}.
Similarly, we prove that \eqref{eq:sumy2b} is implied by \eqref{eq:sumx2b} and \eqref{eq:topology_4}.
\end{proof}

\subsubsection{Simplified formulation}
The results obtained in \ref{subsubsection:reduce_variables} and \ref{subsubsection:reduce_constraints} show that
the GED problem can also be solved by using \eqref{eq:objective2} as the objective function and \eqref{eq:sumx1b},
\eqref{eq:sumx2b}, \eqref{eq:topology_3} and \eqref{eq:topology_4} as the constraints of the BLP. We finally come up with a simplified formulation of the GED problem:

\begin{subequations}
  \begin{center}(F2)\end{center}
  \begin{equation}
      \begin{aligned}
        \min_{\mathbf{x,y}} \Biggl( \sum_{i \in V_1}\sum_{k \in V_2}
        \Bigl(c(i \rightarrow k) - c(i \rightarrow \epsilon)
        - c(\epsilon \rightarrow k)\Bigr) \cdot x_{i,k} \\
        + \sum_{ij \in E_1}\sum_{kl \in E_2} \Bigl(c(ij \rightarrow
        kl) - c(ij \rightarrow \epsilon)
        - c(\epsilon \rightarrow kl)\Bigr) \cdot y_{ij,kl} \\
        + C \Biggr)
      \end{aligned}
    \label{f2:o}
  \end{equation}
  \begin{align}
    \text{subject to}\quad
    &\sum_{k \in V_2} x_{i,k} \leq 1 \quad \forall i \in V_1\label{f2:c1}\\
    &\sum_{i \in V_1} x_{i,k} \leq 1 \quad \forall k \in V_2\label{f2:c2}\\
    &\sum_{kl \in E_2} y_{ij,kl} \leq x_{i,k} \quad \forall k \in V_2, \forall ij \in E_1\label{f2:c3}\\
    &\sum_{kl \in E_2} y_{ij,kl} \leq x_{j,l} \quad \forall l \in V_2, \forall ij \in E_1\label{f2:c4}\\
    \text{with}\quad
    &x_{i,k} \in \{0, 1\} \quad \forall (i, k) \in V_1 \times V_2\label{f2:d1}\\
    &y_{ij,kl} \in \{0, 1\} \quad \forall (ij, kl) \in E_1 \times E_2\label{f2:d2}
  \end{align}
\end{subequations}

The formulation (F2) has:
\begin{itemize}
  \item $|V_1|\cdot|V_2| + |E_1|\cdot|E_2|$ variables,
  \item $|V_1| + |V_2| + 2 |V_2|\cdot|E_1|$ constraints (without the domain constraints).
\end{itemize}

Thus, it uses less variables than (F1), and depending on the density of the graphs, it potentially
uses less constraints to solve the same problem.

\subsection{Solving the programs}
Solving an ILP is NP-hard \cite{garey79:_comput_and_intrac}, thus exploring the entire solution tree is not an option since it would take an exponential time.
However, dedicated solvers have been developed to reduce the number of explored solutions and the solving time, by using a branch-and-cut algorithm along with some heuristics \cite{Wolsey1998}.

Once equations \eqref{ilp:o} to \eqref{ilp:c2} are correctly formulated, the second step consists in implementing this model using a mathematical solver. Given an instance of the problem, the solver explores the tree of solutions with the branch-and-bound algorithm, and finds the best feasible solution, in terms of the objective function optimization.

\subsection{Lower bounding the GED with continuous relaxation}

The common resolution method of an ILP consists in using a branch-and-bound algorithm.
The continuous relaxation of an ILP, i.e. a linear program (LP) where the constraints remain unmodified but where the variables
are now continuous, is a lower bound of the minimization problem that can be solved in polynomial time $\mathcal{O}(n^{3.5})$ with the interior point
method \cite{saigal95:_linear_progr}. This lower bound helps the ILP solving since it allows to prune the exploration of the solution tree.

However, the continuous relaxation can also be used to approximate the optimal objective value in polynomial time. We call F1LP (resp. F2LP) the continuous relaxation of F1 (resp. F2). To this end, we only substitute the discrete space $\{0, 1\}$ by the continuous space $[0, 1]$ in domain constraints (12i) to (12n).


\section{Experiments}
\label{sec:experiments}
As stated in the introduction, one of the contributions of this article is to provide to the reader a robust experimental study through a comparison of eight methods.
In this section, we first describe the methods that have been studied. Then, the datasets and the protocol used to compare the reference methods and our proposals are described. Finally, the results are presented and discussed.

\subsection{Studied methods}
\label{subsec:studiedmethods}
We compare the four approaches proposed previously with four other graph edit distance algorithms from the literature. From the related work, we chose one exact method and three approximate methods. On the exact method side, A$^*$ algorithm applied to GED problem \cite{DBLP:conf/mlg/RiesenFB07} is a foundation work. In our tests, the heuristic is computed thanks to the approximation based on bipartite graph matching. It is the most well-known exact method and it is often used to evaluate the accuracy of approximate methods.
On the approximate method side, we can distinguish three families of methods, tree-based methods, assignment-based methods and set-based methods. For the tree-based methods, a truncated version of A$^*$ called beam search was chosen. This method is known to be one of the most accurate heuristic from the literature \cite{DBLP:conf/sspr/NeuhausRB06}. Among the assignment-based methods, we selected the bipartite graph matching described in \cite{DBLP:journals/ivc/RiesenB09}. In \cite{DBLP:journals/ivc/RiesenB09}, authors demonstrated that this upper bound is a good compromise between speed and accuracy. Finally, we picked a very recent set-based method. In 2014, A. Fischer et al \cite{Fischer2015331} proposed an approach based on the Hausdorff matching. This method is a lower bound of the GED problem.
All these methods cover a large range of GED solvers. In table \ref{optimalnotation}, for each method, acronym, type of method (exact or not) as well as a short synthesis are presented.
We could not assess our methods against all the state of the art. Among the missing methods, we did not compare experimentally our proposals against the binary linear programs proposed by Justice and Hero \cite{Justice:2006:BLP:1155317.1155424}. Despite our best efforts, we could not find the source code of the method or binary files and neither the datasets used in their experiments. 




\begin{center}
 
\begin{table}[!h]
\small

\begin{tabular}{|p{2.5cm}|p{5.5cm}|}
\hline 
\textbf{Acronym / Type} & \textbf{Description of the method} \\ 
\hline 
A* (\cite{DBLP:conf/mlg/RiesenFB07} ) \newline Exact& A* algorithm using a bipartite heuristic.\\ 
\hline 
F1  (\emph{this paper})  \newline Exact & Our first binary linear programming formulation.\\ 
\hline
F2 ( \emph{this paper})  \newline Exact & Our second BLP formulation, derived from (F1).\\ 
\hline
\hline
BP  \cite{DBLP:journals/ivc/RiesenB09} \newline Upper bound & Bipartite graph matching using Munkres algorithm. \\ 
\hline 
BS-$q$ \cite{DBLP:conf/sspr/NeuhausRB06} \newline Upper bound  & A* algorithm with beam search approach and using a bipartite heuristic. \\
\hline 
H  \cite{fischer2013fast} \newline Lower bound  & Modified Hausdorff distance applied to graphs.  \\ 
\hline 
F1LP \emph{this paper} \newline Lower bound  & Linear programming approach, continuous relaxation of (F1).  \\
\hline
F2LP  \emph{this paper} \newline Lower bound  & Linear programming approach, continuous relaxation of (F2).  \\
\hline
\end{tabular} 
\caption{Notations corresponding to each optimal or suboptimal method}
\label{optimalnotation}
\end{table}

\end{center}

\subsection{Datasets}
\label{subsec:datasets}
Graph edit distance algorithms are applied to three different real world graph datasets (GREC, Protein, Mutagenicity) and to one synthetic dataset (ILPISO).
Real world datasets are described in \cite{iamdb} while the synthetic dataset is depicted in \cite{LeBodic20124214}. 
All datasets are publicly available on IAPR Technical commitee $\#$15 website\footnote{\url{https://iapr-tc15.greyc.fr/links.html\#Benchmarking\%20and\%20data\%20sets}}.
From these datasets, we have built subsets where all graphs have the same number of vertices in order to evaluate the algorithms behaviours when complexity grows. The underlying assumption is that the problem becomes more complex as the graphs hold more vertices.
Each dataset is described in three steps. We first present the application field and the graph construction. Secondly, the cost function used for the considered dataset is presented. Finally, the interest of the dataset is discussed.
A synthesis concerning those data are given in table \ref{tab:datasets}.
For each dataset, the corresponding subset and the code of the cost function are available at \url{https://sites.google.com/site/blpged/}.

\subsubsection{GREC dataset (GREC)}
The GREC dataset consists of graphs representing symbols from architectural and electronic drawings. The images occur at five different distortion levels. The result is thinned to obtain lines of one pixel width. Finally, graphs are extracted from the resulting denoised images by tracing the lines from end to end and detecting intersections as well as corners.
Ending points, corners, intersections and circles are represented by vertices and labeled with a two-dimensional attribute giving their position. The vertices are connected by undirected edges which are labeled as line or arc. An additional attribute specifies the angle with respect to the horizontal direction or the diameter in case of arcs.
From the original GREC dataset \cite{GRECdb}, 22 classes are considered. From IAM GREC dataset, subsets were built and their corresponding characteristics is provided in table \ref{tab:subdatasetsPROT}.

\paragraph*{Cost function}
Additionally to (x, y) coordinates, the graph vertices are labeled with a type (ending point, corner, intersection, circle). The same goes with the edges where two types (line, arc) are employed. The Euclidean cost model is adapted accordingly. That is, for vertex substitutions the type of the involved vertices is compared first.
For identically typed vertices, the Euclidean distance is used as vertex substitution cost. In case of non-identical types on the vertices, the substitution cost is set to $2 \cdot \tau_{vertex}$, which reflects the intuition that vertices with different type label cannot be substituted but have to be deleted and inserted, respectively.
For edge substitutions, we measure the dissimilarity of two types with a Dirac function returning 0 if the two types are equal, and $2 \cdot \tau_{edge}$ otherwise. 
Meta-parameters $\tau_{vertex}$ and $\tau_{edge}$ are explained in section parameter settings \ref{sec:paramsetting}.
Elementary operation costs are set up from \cite{Riesen:2010:GCC:1855255} and they are reported in table \ref{tab:metaparameters}.

\paragraph*{Dataset interest}
GREC dataset is composed of undirected graphs of rather small size (i.e. up to 20 vertices in our experiments). In addition, continuous attributes on vertices and edges play an important role in the matching procedure. Such graphs are representative of pattern recognition problems where graphs are involved in a classification stage.

\subsubsection{Mutagenicity dataset (MUTA)}
Mutagenicity is one of the numerous adverse properties of a compound that
hampers its potential to become a marketable drug \cite{doi:10.1021/jm040835a}.
This dataset consists of two classes (mutagen, nonmutagen), which represent molecules.
The molecules are converted into graphs in a straightforward manner by representing atoms as vertices and
the covalent bonds as edges. Vertices are labeled with their chemical symbol and edges by the valence of the linkage.
From this dataset, subsets were generated with the idea to build subfolds where all graphs have the same number of vertices spaced exactly by\textit{ 10 : 10 : 70} vertices. Every subfold holds exactly 10 graphs.

%
%

\paragraph*{Cost function}
Edge substitutions are free of cost. For vertex substitutions, we measure the dissimilarity of two chemical symbols with a Dirac function returning 0 if the two symbols are equal, and 2. $\tau_{vertex}$ otherwise.

\paragraph*{Dataset interest}
This dataset is representative of exact matching problems in the way that a significant part of the topology together with the corresponding vertex and edge labels in $G_1$ and $G_2$ have to be identical. In addition, this set of graphs gathers large instances with up to 70 vertices.
Elementary operation costs are set up from \cite{Riesen:2010:GCC:1855255} and they are reported in table \ref{tab:metaparameters}.

\subsubsection{Protein dataset (PROT)}
The protein dataset contains graphs representing proteins originally used in \cite{4202588}. The graphs are constructed from the Protein Data Bank and labeled with their corresponding enzyme class labels from the BRENDA enzyme dataset \cite{Schomburg:2004nu}.
The protein graphs are split into six classes (EC 1, EC 2, EC 3, EC 4, EC 5, EC 6), which represent proteins out of the six enzyme commission top level hierarchy (EC classes).
The proteins are converted into graphs by representing the secondary structure elements of a protein with vertices and edges of an attributed graph. Vertices are labeled with their type (helix, sheet, or loop) and their amino acid sequence (e.g. TFKEVVRLT). Every vertex is connected with an edge to its three nearest neighbors in space. Edges are labeled with their type and the distance they represent in angstroms. A summary of all subsets and their corresponding characteristics is provided in table \ref{tab:subdatasetsPROT}.

\begin{table}[!h]
\begin{center}
\centering
\begin{tabular}{|p{1cm}|c|c|c||c|c|c||c|c|c|c|}
\hline  
& \multicolumn{3}{c|}{PROT} & \multicolumn{3}{c|}{ILPISO}  &\multicolumn{4}{c|}{GREC}        \\
\hline
$\#$vertices & 20 & 30 & 40& 10 & 25 & 50& 5 & 10 &15  &20   \\ 
\hline 
$\#$graphs & 15 & 13 & 22& 12 & 12 & 12& 41 & 74 & 34 & 39  \\ 
\hline 
\end{tabular}
\caption{Subsets decomposition of PROT, ILPISO and GREC datasets} 
\label{tab:subdatasetsPROT}
\end{center}
\end{table}

\paragraph*{Cost function}
For the protein graphs, a cost model based on the amino acid sequences is used. For vertex substitutions, the type of the involved vertices is compared first. If two types are identical, the amino acid sequences of the vertices to be substituted are compared by means of string edit distance. Similarly to graph edit distance, string edit distance is defined as the cost of the minimal edit path between a source string and a target string. More formally, given an alphabet L and two strings s1, s2 defined on L (s1, s2 $\in$ $L^*$), we allow substitutions, insertions, and deletions of symbols and define the corresponding cost as follows :$$c(u \rightarrow v) = c(u \rightarrow \epsilon) = c(\epsilon \rightarrow v)=1 \text{ for } u,v \in L, u \ne v$$

Hence, vertex substitution cost is defined as the minimum cost sequence of edit operations that has to be applied to the amino acid sequence of the source vertex in order to transform it into the amino acid sequence of the target vertex.
If two vertex types (helix, sheet, or loop) are not identical then the substitution is equivalent to a vertex deletion. For edge substitutions, we measure the dissimilarity with a Dirac function returning 0 if the two edge types are equal, and $2 \cdot \tau_{edge}$ otherwise.
Elementary operation costs are set up from \cite{Riesen:2010:GCC:1855255} and they are reported in table \ref{tab:metaparameters}.

\paragraph*{Dataset interest}
The stringent constraints imposed by exact vertex matching is relaxed thanks to the string edit distance. So the matching process can be tolerant and accommodate with differences.

\subsubsection{ILPISO dataset (ILPISO)}
Four synthetic datasets are provided. Each of them is composed of several triplets (pattern graph, target graph and groundtruth). The graphs have been randomly generated thanks to the Erdos-Renyi model \cite{Erdos60onthe} with or without the constraint of producing connected graphs. For each option, one version of the dataset has an exact mapping (equal labels between matched vertices/edges) whereas an other version includes noise on label values. The groundtruth information gives the one-to-one vertex mapping involving the minimal cost assignment. Each vertex and edge is labelled with a single continuous value in $[-100,+100]$. Edge and vertex attributes follow a uniform law $\mathcal{U}(-100,100)$.
A summary of all subsets and their corresponding characteristics is provided in table \ref{tab:subdatasetsPROT}.

\paragraph*{Cost function}
For vertex substitutions, we measure the dissimilarity of two vertices with an absolute difference. For vertex deletion and insertion, a fixed cost is chosen which is equal to $\dfrac{2}{3}100 \approx 66,6$.
Elementary operation costs are reported in table \ref{tab:metaparameters}.

\paragraph*{Dataset interest}
This dataset stands apart from the others in the sense that this dataset hold directed graphs. The aim is to illustrate the flexibility of our proposal that can handle different type of graphs.

\begin{table}
\centering
\begin{tabular}{|c||p{1.2cm}|p{1.2cm}|p{1.2cm}|p{1.2cm}|}
\hline 
 & MUTA & GREC & PROT &ILPISO \\ 
\hline 
\hline 
Size &4337  & 1100 & 600 & 36\\ 
\hline 
Vertex labels &Chemical symbol & x, y coordinates & Type and aa-sequence &scalar value\\ 
\hline 
Edge labels & Valence & Line type & Type and distance &scalar value\\ 
\hline 
$\overline{vertices}$ & 30.3 & 11.5 &32.6& 28.3\\ 
\hline 
$\overline{edges}$ &30.8  & 12.2 & 62.1 &54.3\\ 
\hline 
Graph type & undirected & undirected & undirected &directed\\ 
\hline 
\end{tabular}
\caption{Summary of the graph datasets characteristics} 
\label{tab:datasets}
\end{table}

\begin{table}
\centering
\begin{tabular}{|c||c|c|c|p{1.5cm}|p{1.5cm}|}
\hline 

& $\tau_{vertex}$ &$\tau_{edge}$ & $\alpha$ & Vertex substitution function & Edge substitution function \\ 
\hline 
\hline 
GREC & 90 & 15 & 0.5 & Extended euclidean distance & Dirac function \\ 
\hline 
PROT & 11 & 1 & 0.75 & Extended string edit distance & Dirac function \\ 
\hline 
MUTA & 11 & 1.1 & 0.25 & Dirac function & Dirac function \\ 
\hline 
ILPISO & 66.6 & 66.6 & 0.5 & L1 norm & L1 norm \\ 
\hline 
\end{tabular}  
\caption{Cost function meta parameters for the four datasets}
\label{tab:metaparameters}
\end{table}

\subsection{Protocol}
In this section, the experimental protocol is detailed. We explain how the experiments were performed and the reasons why we led these tests.

Our experiments were carried out in a context of graph comparisons. Let $\mathcal{S}$ be a graph dataset consisting of $m$ graphs, $\mathcal{S} = \{G_1, G_2, ..., G_m\}$. Let $\mathcal{P} = \mathcal{P}_e \cup \mathcal{P}_a$ be the set of all graph edit distance methods listed in \ref{subsec:studiedmethods}, with $\mathcal{P}_e = \{\text{A*, F1, F2}\}$ the set of exact methods and $\mathcal{P}_a = \{\text{BP, BS-10, H, F1LP, F2LP}\}$ the set of approximate methods (see table \ref{optimalnotation} for notations).
Given a method $p \in \mathcal{P}$, we computed the square distance matrix $M^p \in \mathcal{M}_{m\times m}(\mathbb{R}^+)$, that holds every pairwise comparison
$M^p_{i,j}=d_p(G_i,G_j)$, where the distance $d_p(G_i,G_j)$ is the value returned by the method $p$ on the graph pair $(G_i,G_j)$ within
a certain time limit, and using the cost metaparameters defined in table \ref{tab:metaparameters}. For instance $M^\text{F1}$ and $M^\text{BP}$ denote distance matrices computed with F1 and BP methods respectively.

Due to the large number of matchings considered and the exponential complexity of the algorithms tested, we allowed a maximum of \textbf{300 seconds}
for any distance computation. 
 When time limit is over, the best solution found so far is outputted by the given method. 
This time constraint is large enough to let the methods search deeply into the solution space and to ensure that many nodes will be explored. The key idea is to reach the optimality whenever it is possible or at least to get as close as possible to the \textit{Graal}, the optimal solution.
This constraint on the system is well admitted in the operational research field \cite{DBLP:journals/scheduling/BaptisteCGT10, Justice:2006:BLP:1155317.1155424}.




Based on this context of pairwise graph comparison, a set of metrics is defined to measure the accuracy and the speed of our four proposed methods and four standard methods.
   

In the next subsections, performance evaluation metrics as well as the experimental settings are detailed.

\subsubsection{Accuracy metrics}
To illustrate the error committed by approximated methods over exact methods, we measure an index called deviation which is defined by equation \ref{equ:dev}.
\begin{equation}
\label{equ:dev}
\textit{deviation}(i,j)^p = \dfrac{\vert M^p_{i,j}-R_{i,j} \vert}{R_{i,j}}, 
 \forall (i,j) \in \llbracket 1,m \rrbracket ^2, 
 \forall p \in \mathcal{P}
\end{equation}
Where $R$ is defined in equation \ref{equ:ref}.

\begin{equation}
\label{equ:ref}
\textit{R}_{i,j} = \min_{p \in \mathcal{P} \setminus \{ F1LP, F2LP, H \} }\{ M^p_{i,j}\}, \; \forall (i,j) \in \llbracket 1,m \rrbracket ^2
\end{equation}
For each comparison, the reference matrix holds the optimal graph edit distance whenever it is possible to compute it. The optimality may not be reached due to time restriction. When no optimal solutions were available, the lowest graph edit distance found among all the methods is chosen to be the reference value. The lower bounds (H, F1LP and F2LP) are removed from the formula \ref{equ:ref} since they do not represent feasible solutions and they cannot represent real sequences of edit operations. 
For a given method, the deviation can express the error made by a suboptimal solution in percentage of the best solution.

For each subset, the mean deviation is derived as follows in equation \ref{equ:meandev} :
\begin{equation}
\label{equ:meandev}
\overline{\textit{deviation}^p}=\dfrac{1}{m \times m} \sum_{i=1}^{m} \sum_{j=1}^{m} \textit{deviation}(i,j)^p
\end{equation}

To obtain comparable results between datasets, mean deviations are normalized between $[0,1]$ as follows in equation \ref{equ:devscore} :
\begin{equation}
\label{equ:devscore}
\overline{\textit{deviation score}^p}=\dfrac{1}{\#subsets} \sum_{i=1}^{\#subsets} \dfrac{\overline{\textit{deviation}^p_i}}{maxdev_i}
\end{equation}

\begin{equation*}
\textit{maxdev}_{i} = \max \overline{\textit{deviation}_i^p}\; \forall p \in \mathcal{P} 
\end{equation*}

Deviation score is a type of measurement used to compare performance over subsets.

\subsubsection{Speed metrics}
To evaluate the convergence of algorithms, the mean time for each dataset is derived as follows in equation \ref{equ:meantime} :
\begin{equation}
\label{equ:meantime}
\overline{\textit{time}^p}=\dfrac{1}{m \times m} \sum_{i=1}^{m} \sum_{j=1}^{m} time(p, G_i, G_j) \;\textrm{and} \; (i,j) \in \llbracket 1, m \rrbracket ^ 2
\end{equation}

Finally, we introduce a last metric called speed score. To compare speed performance over datasets, the running time is normalized between $[0,1]$ as follows in equation \ref{equ:speedscore} :
\begin{equation}
\label{equ:speedscore}
\overline{\textit{speed score}^p}=\dfrac{1}{\#subsets} \sum_{i=1}^{\#subsets} \dfrac{\overline{\textit{time}^p_i}}{maxtime_i}
\end{equation}
\begin{equation*}
\textit{maxtime}_{i} = \max \overline{\textit{time}_i^p}\; \forall p \in \mathcal{P} 
\end{equation*}
These evaluations were run on datasets GREC, PROT, MUTA and ILPISO. In order to show the impact of the graph size on the problem complexity, we performed
our experiment on subsets where all graphs have the same number of vertices.

\subsubsection{Experimental settings}
\label{sec:paramsetting}
For the understanding of these tests, we first recall notations that will make the reading much simpler. Graph edit distance holds meta parameters which are domain-depend costs. We borrow notations from Kaspar Riesen thesis report \cite{Riesen:2010:GCC:1855255}. $\tau_{node}$ corresponds to the cost of a node deletion or insertion, $\tau_{edge}$ corresponds to the cost of an edge deletion or insertion, $\alpha \in [0,1]$ corresponds to the weighting parameter that controls whether the edit operation cost on the nodes or on the edges is more important. Elementary operation costs are reported in table \ref{tab:metaparameters}.

In this practical work, the $BP$ was provided by the Institute of Computer Science and Applied Mathematics of Bern in Switzerland\footnote{\url{http://www.iam.unibe.ch/fki/}}, while other methods were re-implemented by
us from the literature. All methods are implemented in JAVA 1.7 except for the F1 and F2 models that are implemented in $C\#$ using CPLEX Concert Technology. CPLEX 12.6 was chosen since it is known to be one of the best mathematical programming solvers. All the methods were run on a 2.6 GHz quad-core computer with 8 GB RAM. For the sake of comparison, none of the methods were parallelized and CPLEX was set up in a deterministic manner.

\subsection{Results}
In this section, we present the results obtained from the experiments.

In figure \ref{meandeviation300s}, the mean deviations of exact methods and approximate methods are presented. Note that A* method was only computed on GREC dataset due its inherent and intractable time complexity. A*'s experiments could not be conducted for graphs larger than 15 vertices with a memory constraint of 1 GB. 
From figure \ref{meandeviation300s} several conclusions can be drawn : 
On all datasets formulation F2 outperforms formulation F1 in terms of accuracy. The gap between both methods can reach 20\% on MUTA dataset. 
Among the lower bounds, F2LP is the most accurate. However lower bounds results are very data dependent. On GREC dataset, the error committed is less than 5\% while on MUTA dataset errors can reach 30\%.
A straightforward remark is that GREC seems to be a quite affordable dataset while MUTA is more challenging. In GREC dataset, solving F2LP leads to near-optimal solutions. In the linear programming formulations, topological constraints of the models are easy to be satisfied. 
The vertices matching constraints (Eq \ref{eq:sumx1}) of having one vertex of $G_1$ matched to only one vertex of $G_2$ fall apart. Solving the continuous relaxation with continuous variables lead to a multivalent matching. The quality of the solution is then mainly supported by the objective function.
The objective function helps at guiding the exploration of the search space. This strengthen the fact that attributes are meaningful and play a more important role than the topology in GREC. Among the methods from the literature, BS is the most accurate except on ILPISO dataset. This comment can be explained due to the directed aspect of the graph involved. Directed edges can be seen as more stringent constraints on the topology. Topology may impact significantly the first branching decisions of the beam search algorithm at the expense of attributes. Beam search back tracking capability is reduced by truncation of the search search space that prevents the method to get back on better branches. 
 A* is probably the worst method when graphs are larger than 10 vertices its error becomes very high (i.e more than 30\%). A* cannot converge to the optimality because of memory saturation phenomenon. The list OPEN containing pending solutions to be expanded grows exponentially according to the graph size. The bipartite heuristic fails to prune the search tree efficiently. To conclude on deviation, the bigger the graphs, the higher the error made by all the methods. Approximate methods may work poorly in cases where neighborhoods do not allow to easily differentiate the partial solutions. Among all approximate methods, F2LP is the most accurate. In average, 6\% more accurate than the second best approximate method which is BS.

\begin{figure}
  \centering
     \subfigure[GREC]{\includegraphics[width=.22\textwidth]{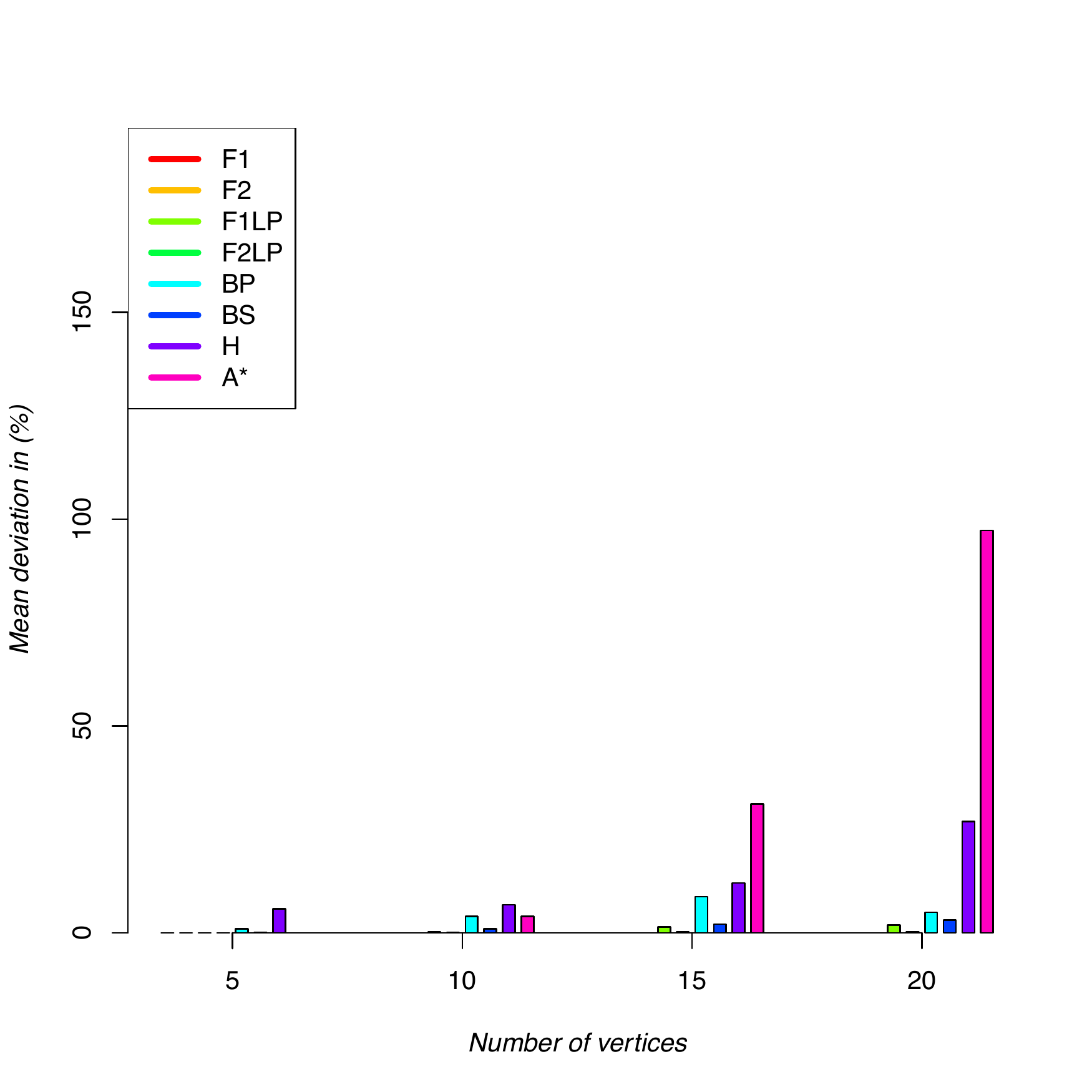}}
    \subfigure[MUTA]{\includegraphics[width=.22\textwidth]{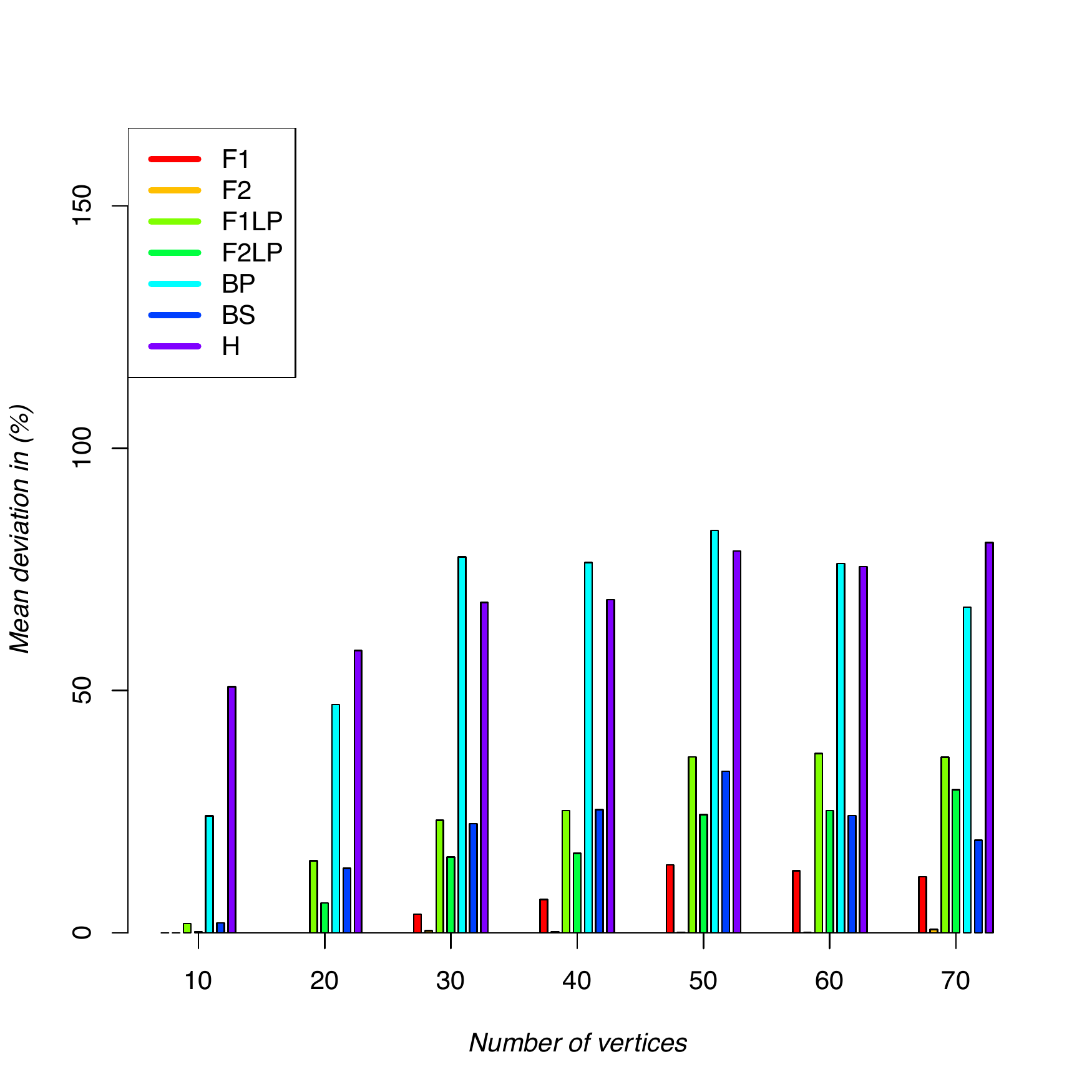}}
 \subfigure[PROT]{\includegraphics[width=.22\textwidth]{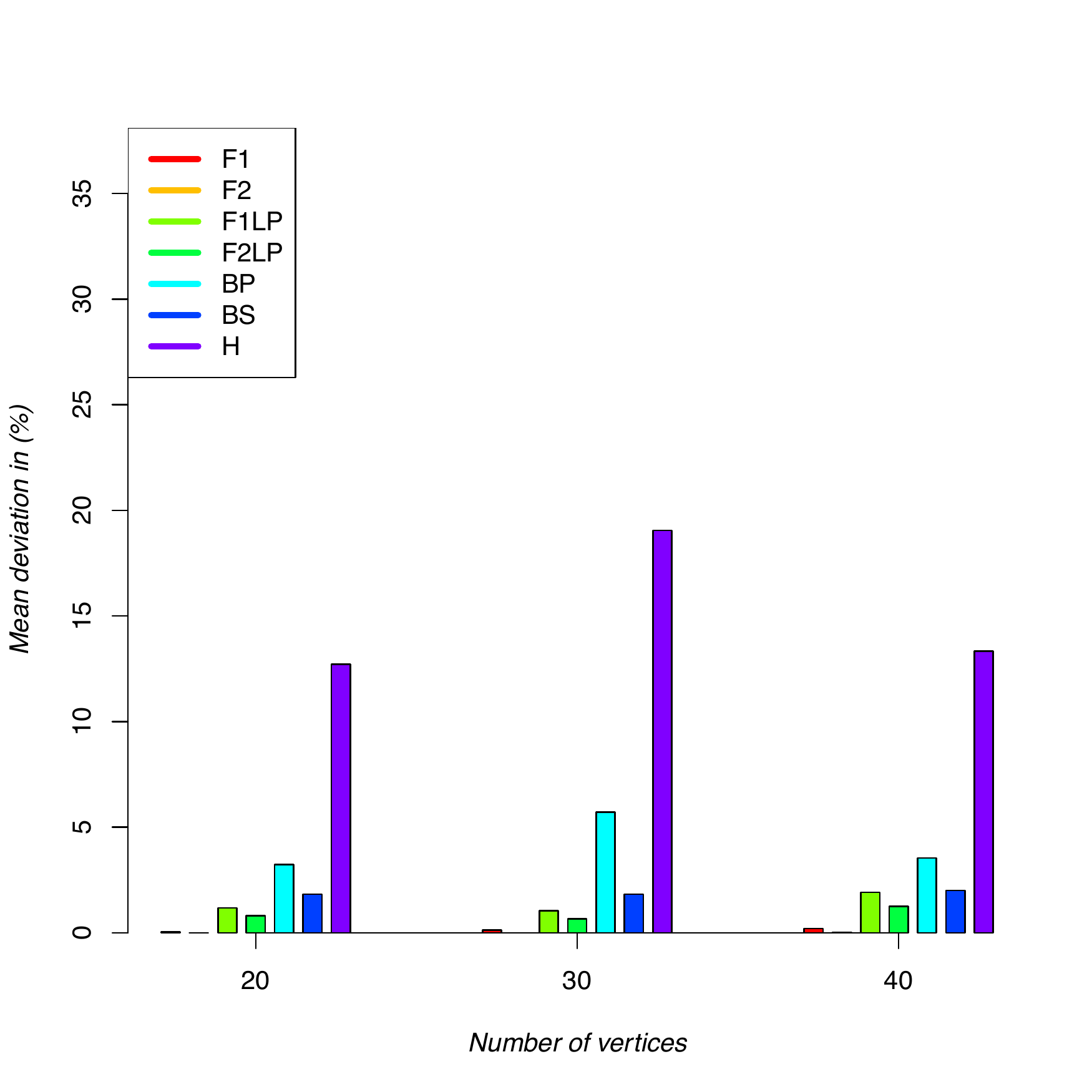}} 
 \subfigure[ILPISO]{\includegraphics[width=.22\textwidth]{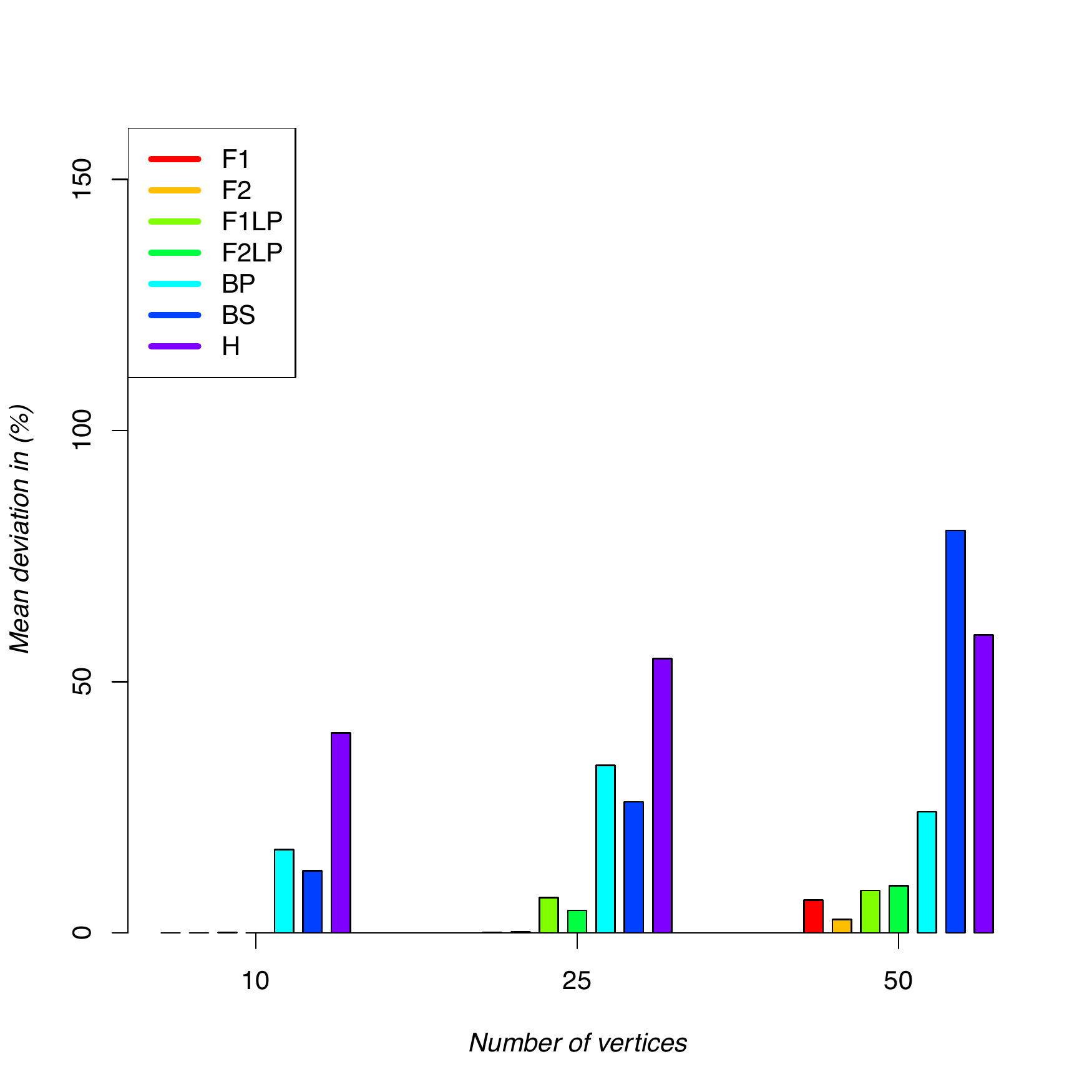}} 
    
 \caption{Mean deviations}
	\label{meandeviation300s}
\end{figure}
In figure \ref{meantime300s}, the average time to compute a graph comparison is depicted for each method. Between F1 and F2 formulations, F2 is always faster than F1. On GREC, F2 can be 100 times faster than F1. On the other hand, the bigger the graphs, the tighter the difference is. In fact, as graphs get bigger, more time is required to solve the problem. When the graph size exceeds 30 vertices the speed of both formulations tends to be similar and the time limit is reached. However in the meantime, F2 would reach a better solution. H and BP are by far the fastest methods. The speed gap with formulation F2 can reach a factor 1000 on MUTA set when graphs get larger. On the other hand, at the scale of exact methods speed, H and BP provide comparable speed results. Finally, A* is the slowest method due to its intensive use of dynamic memory allocation, the best-first search and a misleading bipartite heuristic. 

\begin{figure}
  \centering
    \subfigure[GREC]{\includegraphics[width=.22\textwidth]{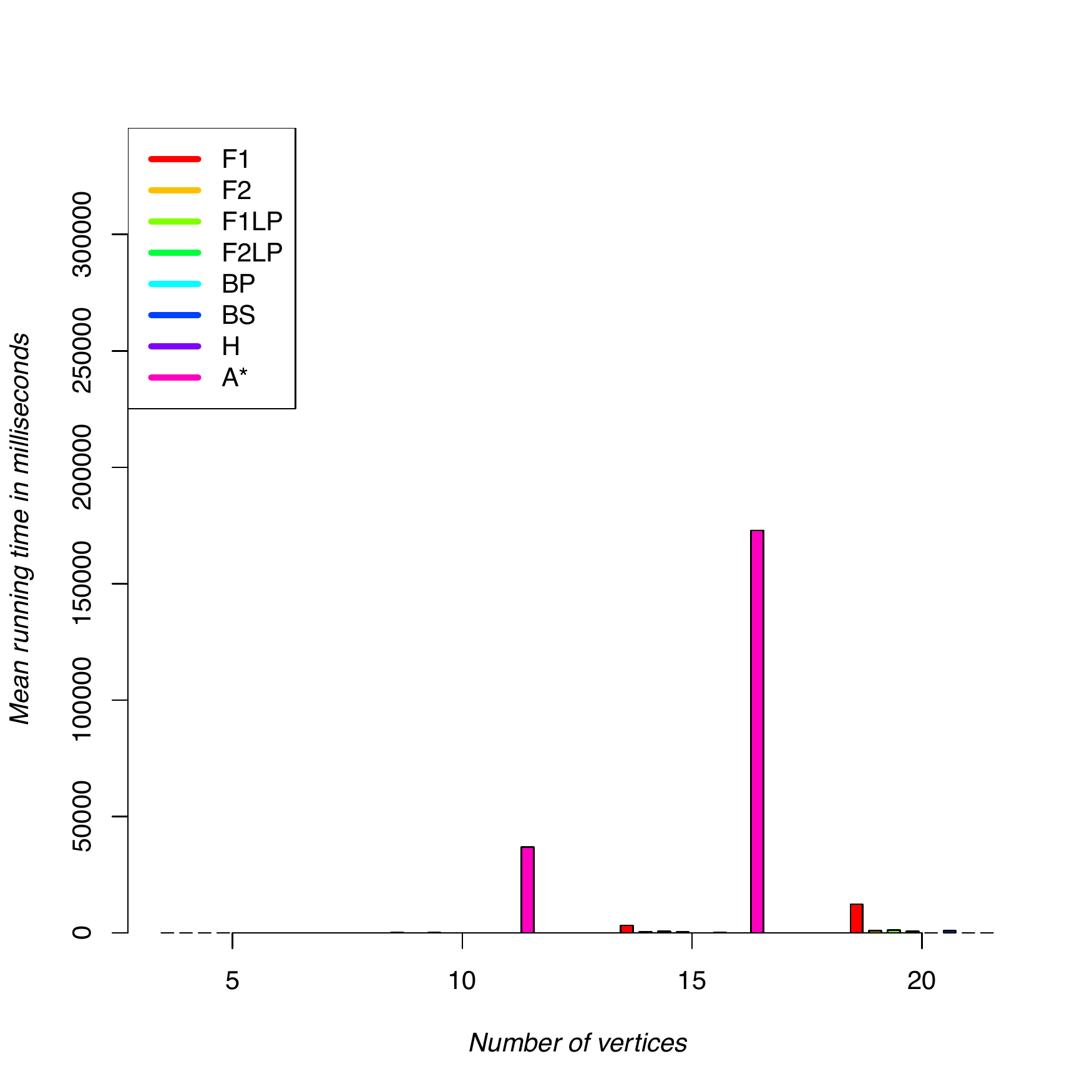}}
    \subfigure[MUTA]{\includegraphics[width=.22\textwidth]{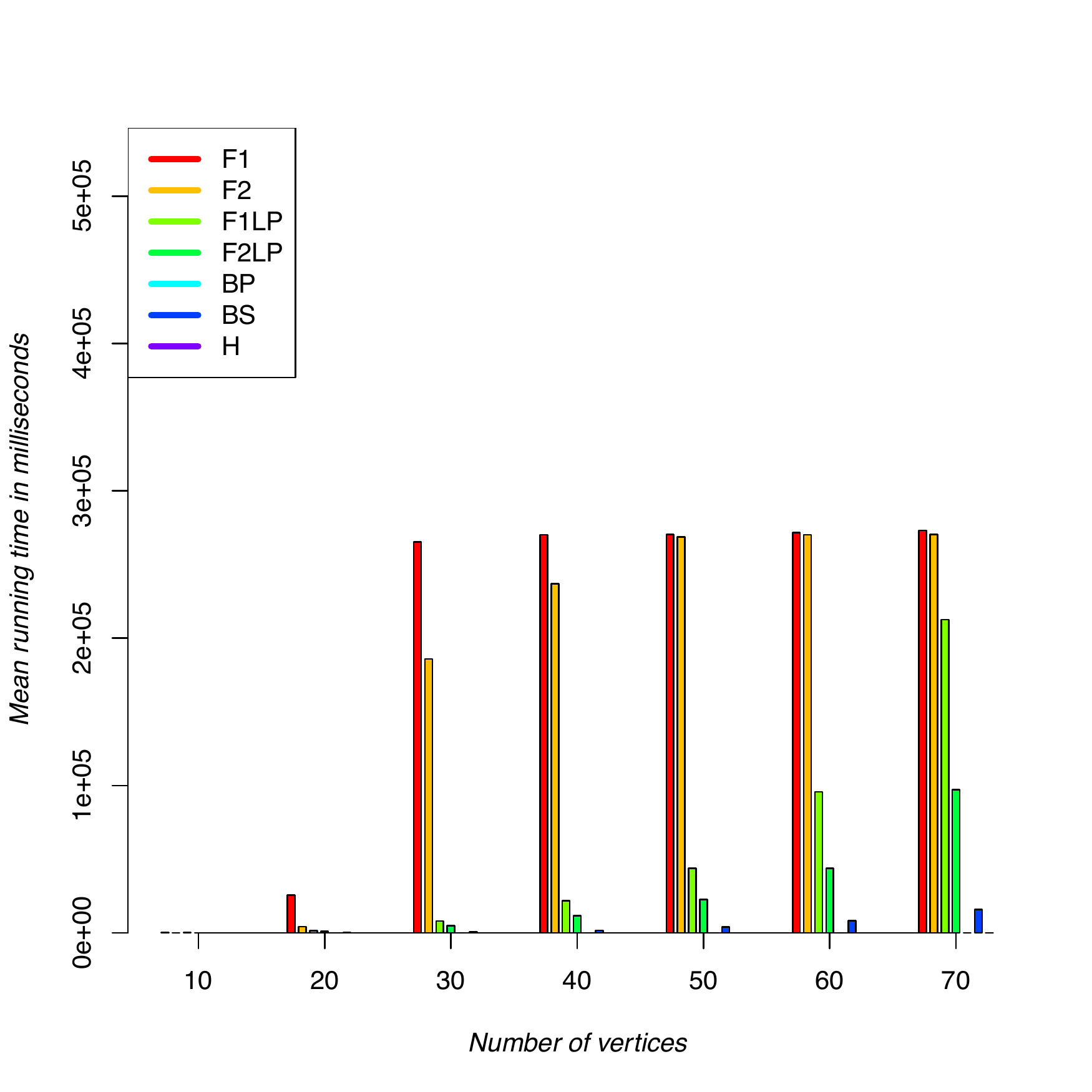}}
 \subfigure[PROT]{\includegraphics[width=.22\textwidth]{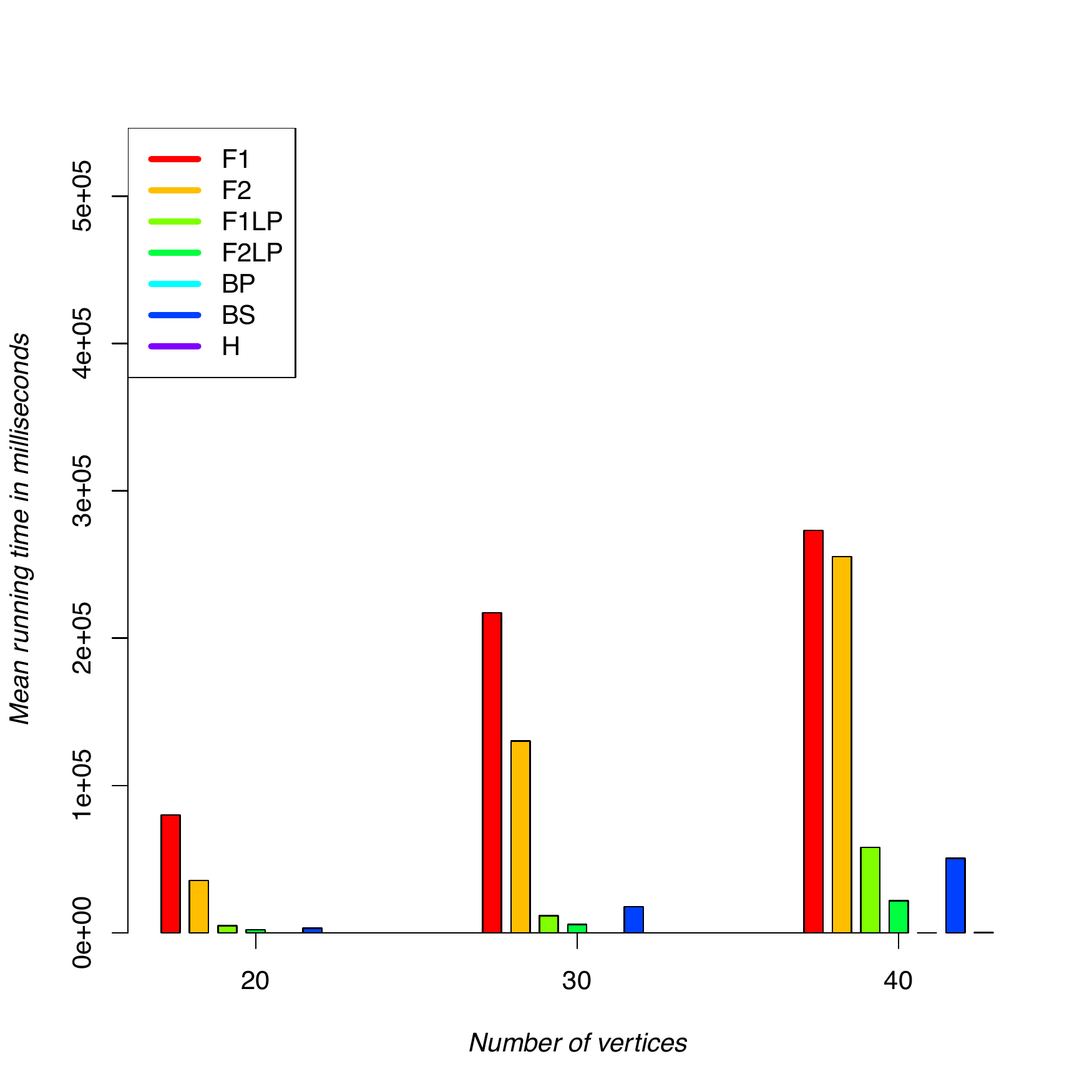}} 
  \subfigure[ILPISO]{\includegraphics[width=.22\textwidth]{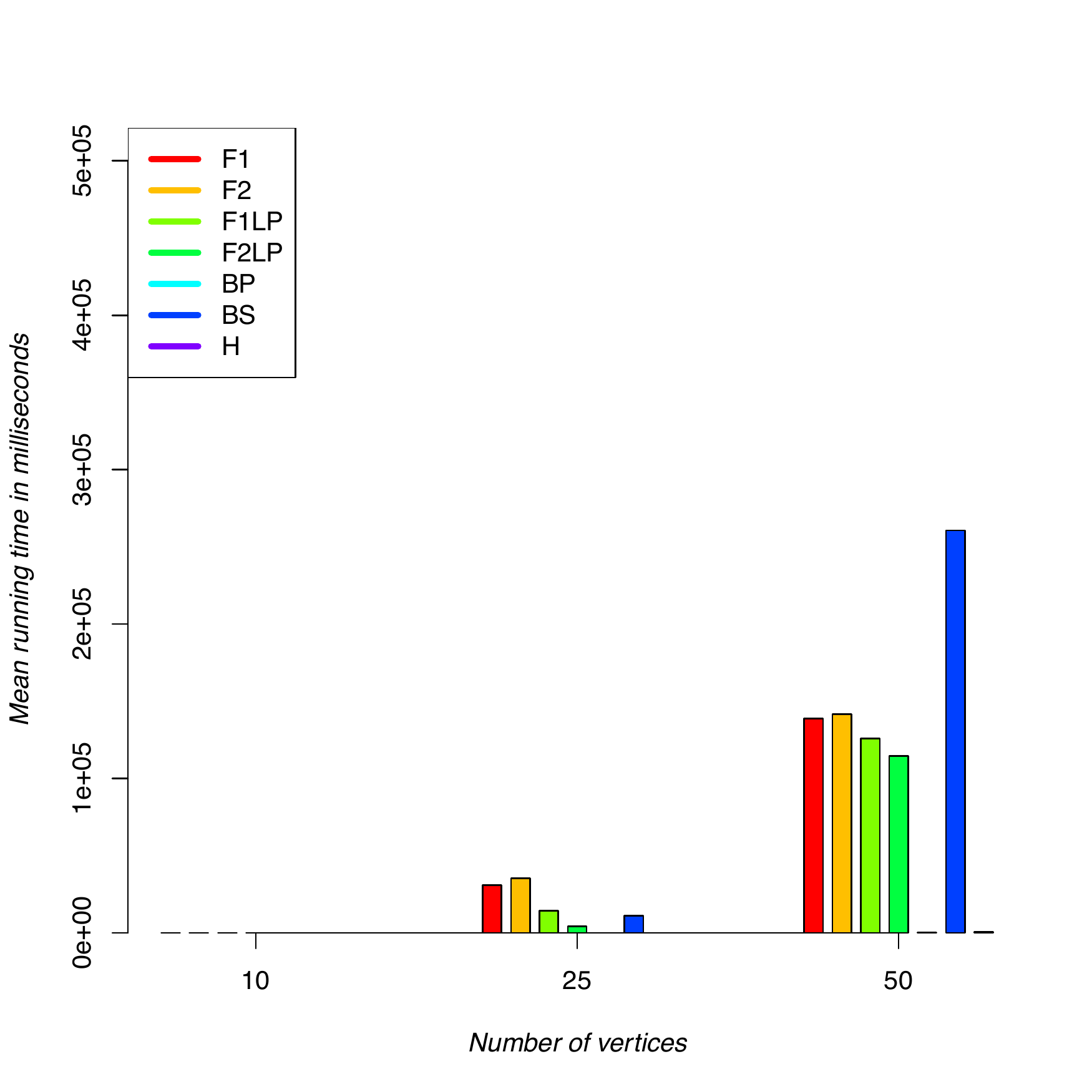}} 
    
 \caption{Mean computation times}
	\label{meantime300s}
\end{figure}


\begin{figure}
  \centering
    \subfigure[GREC]{\includegraphics[width=.22\textwidth]{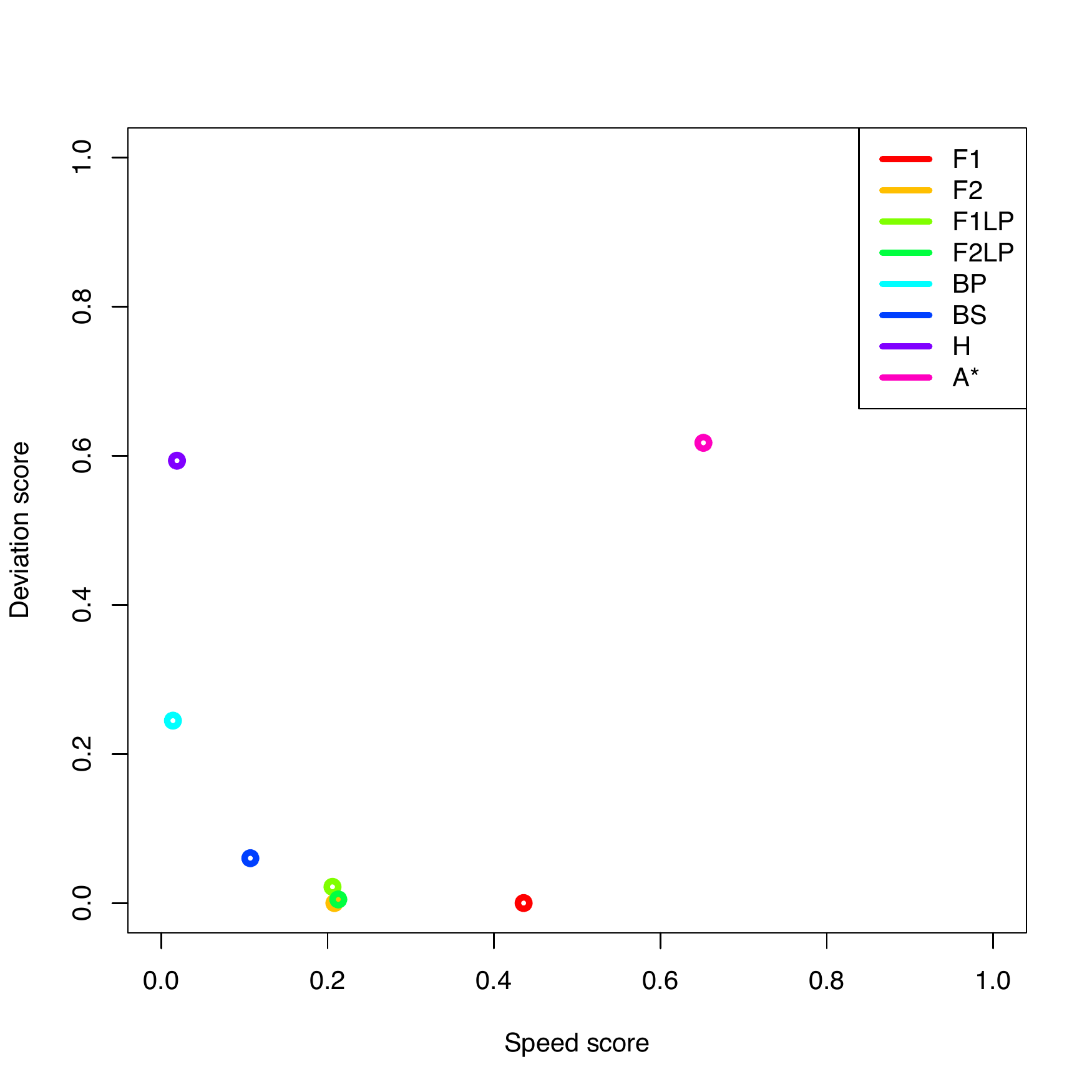}}
    \subfigure[MUTA]{\includegraphics[width=.22\textwidth]{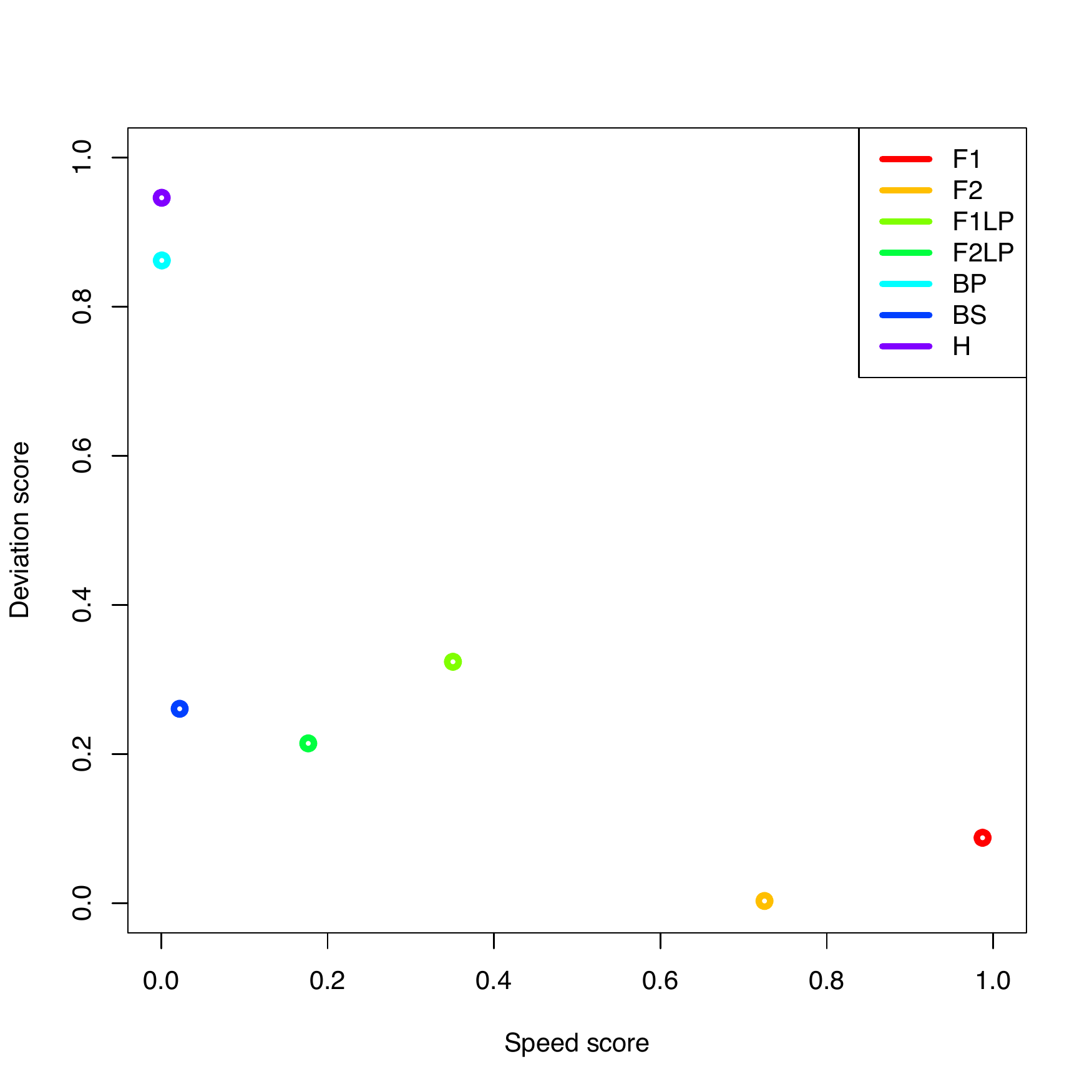}}
    \subfigure[PROT]{\includegraphics[width=.22\textwidth]{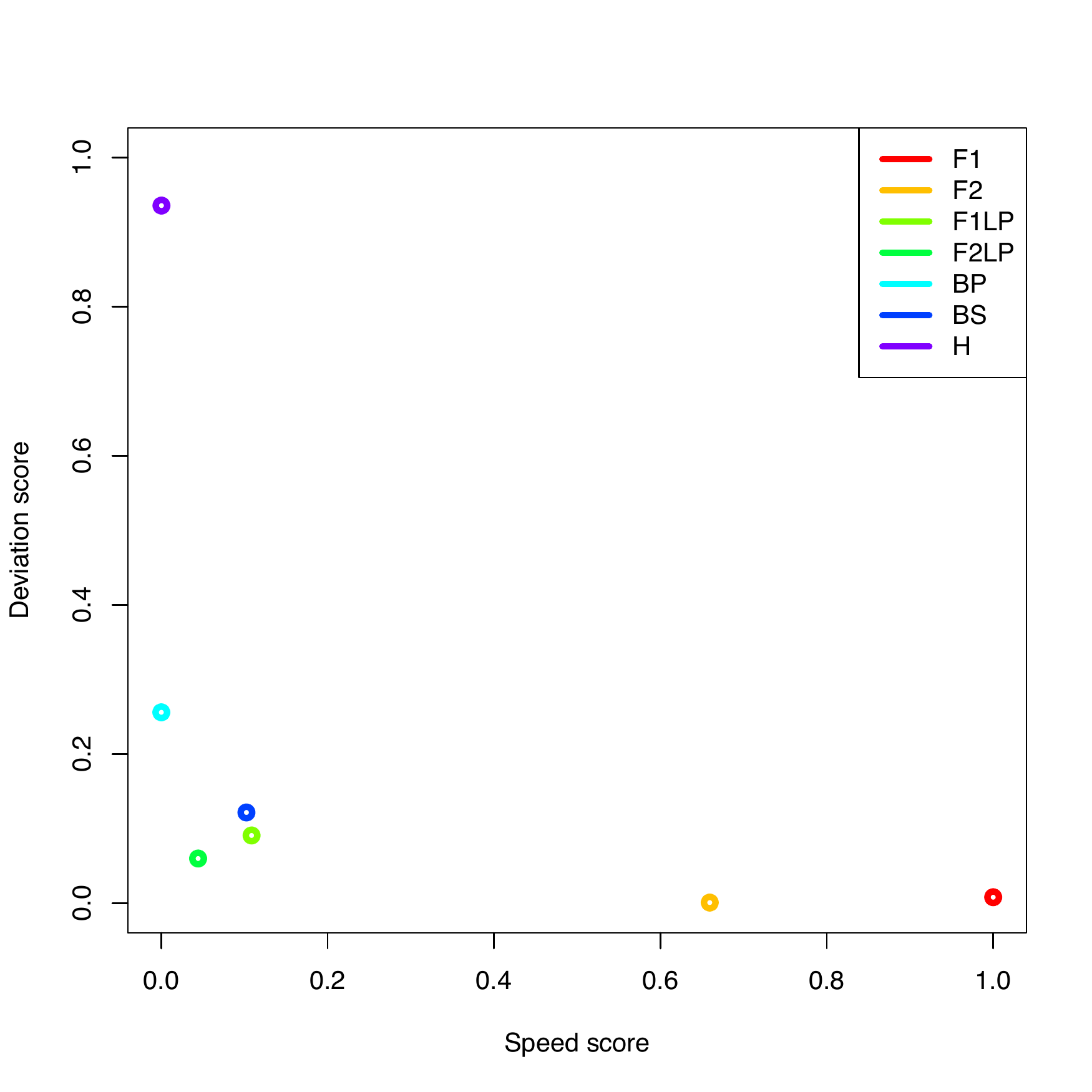}}
        \subfigure[ILPISO]{\includegraphics[width=.22\textwidth]{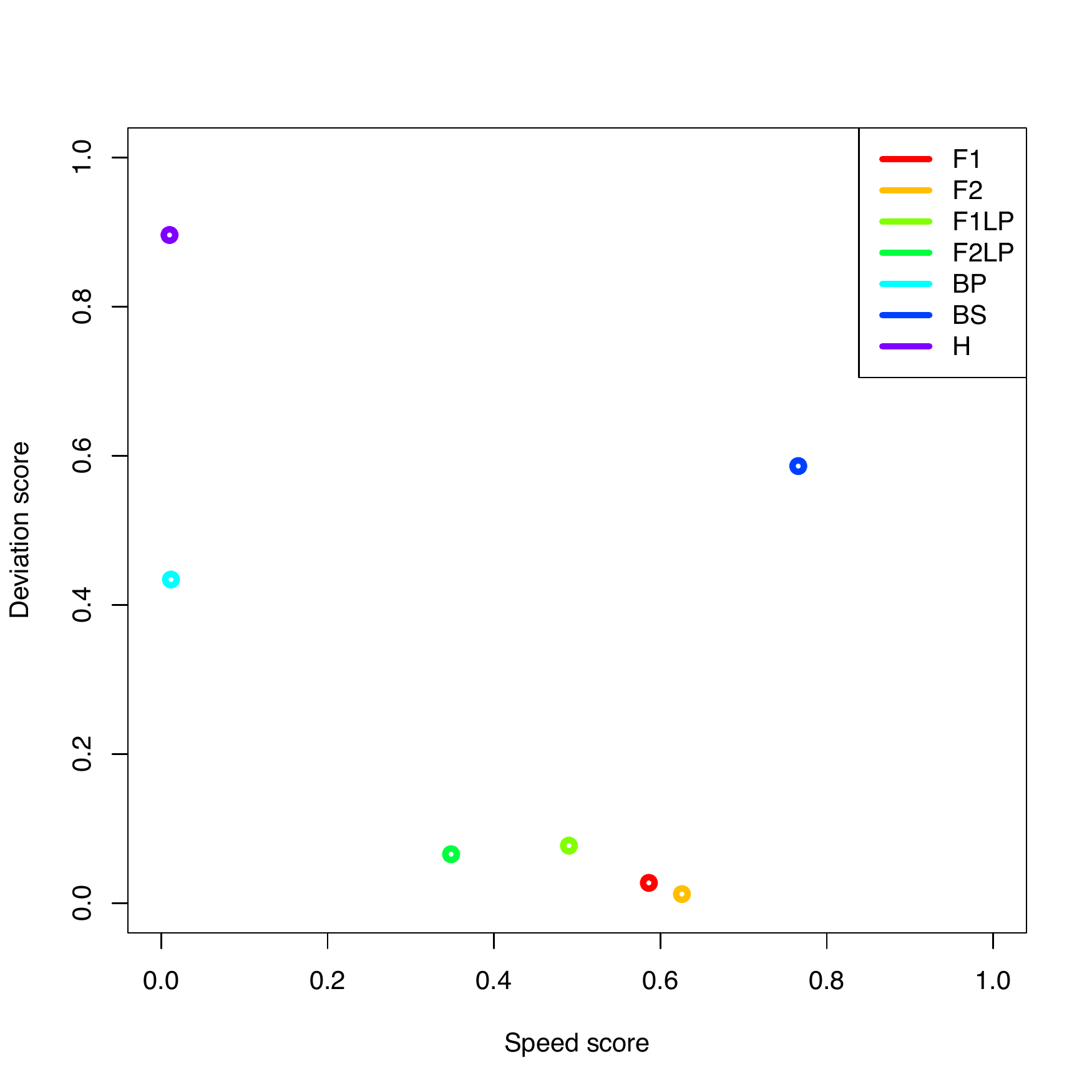}}
     \subfigure[ALL datasets]{\includegraphics[width=.22\textwidth]{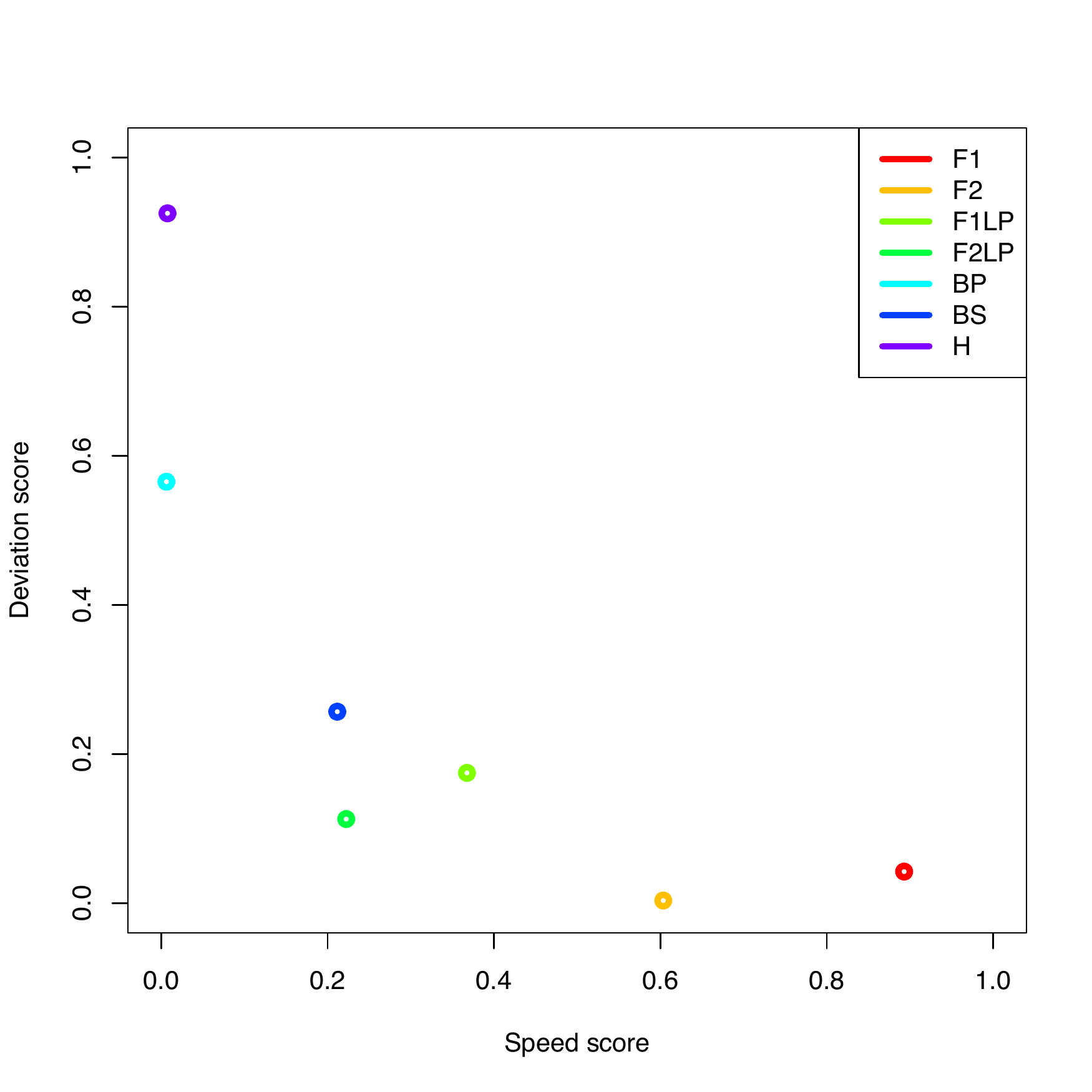}}
 \caption{A synthesis on deviation and time complexity. Lowest deviation and speed time scores are the best. Sub-figure \textit{e} is obtained by merging GREC, PROT, MUTA and ILPISO results.}
	\label{speedaccuracy300s}
\end{figure}

To sum up advantages and drawbacks, each method is projected on a two-dimensional space ($\mathbb{R}^2$) by using speed score and deviation score features defined in equations \ref{equ:speedscore} and \ref{equ:devscore}. Speed and deviation are two concurrent criteria to be minimized. Figure \ref{speedaccuracy300s} illustrates the methods projected in the speed-deviation space. A* method can be categorized as a dominated method since it does not outperform any other methods on either deviation or speed criterion. 
Methods behave differently according to the datasets. \emph{There's no such thing as a free lunch}\footnote{A quote from the economist Milton Friedman} : error-tolerant matching is an NP-hard problem and no methods can fit all problems. 
A quantitative analysis is proposed in tables \ref{tab:deviationgap} and \ref{tab:timefactor}. In table \ref{tab:deviationgap}, deviation gaps between methods are presented while in table \ref{tab:timefactor}, the time ratio between methods are depicted.
Generally speaking, mathematical models seems to be quite accurate and outperform in this way other methods from the literature. F2 outperforms the other methods on all datasets in term of accuracy.
Among approximate method, F2LP is the most precise heuristic. The gap between F2 and F2LP is about 7\%. The most challenging conventional method is BS which is 5 times faster than F2 but in average over all the datasets, F2 is 15\% more accurate than BS. 
On the other hand, BP and H are the fastest methods and any instance can be solved in less than three seconds. Among approximate methods, F2LP is 8\% more accurate than BS in average the reverse side of the medal is an extra amount of time of 13\% in average.


\begin{table}
\begin{tabular}{|c||c|c|c|c|c|c|c|}
\hline 
\% &F1&F2&F1LP&F2LP&BP&BS&H\\ \hline  
\hline  
F1&0&3&8.3&4.6&29.5&12.6&39.7\\ \hline  
F2&3&0&11.4&7.7&32.5&15.7&42.8\\ \hline  
F1LP&8.3&11.4&0&3.7&21.1&4.3&31.4\\ \hline  
F2LP&4.6&7.7&3.7&0&24.8&8&35.1\\ \hline  
BP&29.5&32.5&21.1&24.8&0&16.8&10.3\\ \hline  
BS&12.6&15.7&4.3&8&16.8&0&27.1\\ \hline  
H&39.7&42.8&31.4&35.1&10.3&27.1&0\\  
\hline 
\end{tabular} 
\caption{Mean deviation gap between methods in percentage over all the databases.}
\label{tab:deviationgap}
\end{table}

\begin{table}
\begin{tabular}{|c||c|c|c|c|c|c|c|}
\hline 
/&F1&F2&F1LP&F2LP&BP&BS&H\\ \hline  \hline 
F1&1&0.86&0.28&0.16&0&0.18&0\\ \hline  
F2&1.16&1&0.33&0.18&0&0.2&0\\ \hline  
F1LP&3.55&3.06&1&0.55&0&0.62&0\\ \hline  
F2LP&6.44&5.54&1.81&1&0&1.13&0\\ \hline  
BP&4648.3&4002&1309&721.8&1&817.06&1.8\\ \hline  
BS&5.69&4.9&1.6&0.88&0&1&0\\ \hline  
H&2514.2&2164.7&708&390.4&0.54&441.9&1\\ \hline  

\end{tabular} 
\caption{Mean time factor between methods over all the databases.}
\label{tab:timefactor}
\end{table}

\section{Conclusion}
\label{sec:conclusion}
In this paper, two exact binary linear programming formulations of the graph edit distance problem have been presented. Both formulations can deal with wide range of attributed relational graphs : directed or undirected graphs, simple graphs or multigraphs, with a combination of symbolic, numeric and/or string attributes on vertices and edges. The first formulation (F1) is a didactic expression of the GED problem, while (F2) is a more refined program where variables and constraints have been condensed to reduce the search space.
From both exact models, two lower bounds (F1LP) and (F2LP) have been derived by continuous relaxation of binary variables. Models were solved by CPLEX solver based on branch-and-cut techniques and the interior point method. Formulations were evaluated on four publicly available databases. In all cases, (F1) and (F1LP) are slower and less accurate than (F2) and (F2LP) respectively. This result validates (F2) and the choice of reducing the number of variables and constraints. (F2) is 15\% more accurate in average than the best method from the literature.
Among approximate methods, (F2LP) is 8\% more accurate than BeamSearch (BS) in average the reverse side of the medal is an extra amount of time of 13\% in average. To take the stock, the choice of a method to solve a problem is a trade-off between speed and accuracy.
In perspective, quadratic programming solvers are getting more and more efficient and we want to investigate the definition of binary quadratic programming formulations of the graph edit distance problem. Finally, another interesting work will be to use lower and upper bounds to build an optimized nearest neighbour search.

\begin{appendices}
\section{Extension to undirected graphs}
\label{appendix:undirected}

Suppose that $G_{1}$ and $G_{2}$ are undirected graphs, i.e. their edges have no orientation. The notations $ij$ and $ji$ refer to the same edge of $E_1$, so do $kl$ and $lk$ in $E_2$. This new assumption leads to revise the constraints of the formulation given for directed graphs.

Considering that (F2) has been shown to be more effective than (F1), we only give (F2u), the formulation dedicated to compute graph edit distance between undirected graphs adapted from (F2). The modifications consist in rewritting the sets of constraints \eqref{eq:topology_3} and \eqref{eq:topology_4} into \eqref{eq:topology_3u}. Indeed, given an edge $ij \in E_1$ and a vertex $k \in V_2$
, there is at most one edge incident to $k$ 
that can be matched to $ij$. Moreover $x_{i,k}$ and $x_{j,k}$ 
can not be simultaneously equal to 1, so the sum $x_{i,k} + x_{j,k}$ is at most equal to 1.

\begin{subequations}
\begin{center}
(F2u)
\end{center}
\begin{equation}
\begin{aligned}
  \min_{\mathbf{x,y}} \Biggl( &\sum_{i \in V_1}\sum_{k \in V_2} \Bigl(c(i \rightarrow k) - c(i \rightarrow \epsilon) - c(\epsilon \rightarrow k)\Bigr) \cdot x_{i,k} \\
			    + &\sum_{ij \in E_1}\sum_{kl \in E_2} \Bigl(c(ij \rightarrow kl) - c(ij \rightarrow \epsilon) \\ 
			    & - c(\epsilon \rightarrow kl)\Bigr) \cdot y_{ij,kl} + C \Biggr)
\end{aligned} 
\end{equation}
subject to
\begin{align}
  &\sum_{k \in V_2} x_{i,k} \leq 1 \quad \forall i \in V_1\\
  &\sum_{i \in V_1} x_{i,k} \leq 1 \quad \forall k \in V_2\\
  &\sum_{kl \in E_2} y_{ij,kl} \leq x_{i,k} + x_{j,k} \quad \forall k \in V_2, \forall ij \in E_1 \label{eq:topology_3u}
\end{align}
with
\begin{align}
  &x_{i,k} \in \{0, 1\} \quad \forall (i, k) \in V_1 \times V_2\\
  &y_{ij,kl} \in \{0, 1\} \quad \forall (ij, kl) \in E_1 \times E_2
\end{align}
\end{subequations}

\end{appendices}

\section*{Acknowledgements}
This work was supported by grants from the ANR-11-JS02-010 Lemon. The authors would like to thank Pierre Le Bodic for its contributions to this work.

\bibliographystyle{unsrt} 
\bibliography{biblio}

\begin{thebibliography}{10}

\bibitem{graphMining_Zou}
Zhaonian Zou, Jianzhong Li, Hong Gao, and Shuo Zhang.
\newblock Mining frequent subgraph patterns from uncertain graph data.
\newblock {\em Knowledge and Data Engineering, IEEE Transactions on},
  22(9):1203--1218, Sept 2010.

\bibitem{graphClusteringChen_Djidjev}
Hristo~N. Djidjev and Melih Onus.
\newblock Scalable and accurate graph clustering and community structure
  detection.
\newblock {\em Parallel and Distributed Systems, IEEE Transactions on},
  24(5):1022--1029, May 2013.

\bibitem{graphClassification_Wu}
J.~Wu, S.~Pan, X.~Zhu, and Z.~Cai.
\newblock Boosting for multi-graph classification.
\newblock {\em Cybernetics, IEEE Transactions on}, PP(99):1--1, 2014.

\bibitem{Kernel_Foggia}
Pasquale Foggia and Mario Vento.
\newblock Graph embedding for pattern recognition.
\newblock In {\em Proceedings of the 20th International Conference on
  Recognizing Patterns in Signals, Speech, Images, and Videos}, ICPR'10, pages
  75--82, Berlin, Heidelberg, 2010. Springer-Verlag.

\bibitem{Kernel_Raviv}
Dan Raviv, Ron Kimmel, and Alfred~M. Bruckstein.
\newblock Graph isomorphisms and automorphisms via spectral signatures.
\newblock {\em IEEE Transactions on Pattern Analysis and Machine Intelligence},
  35(8):1985--1993, 2013.

\bibitem{ExplicitEmbedding_Kramer}
Stefan Kramer and Luc~De Raedt.
\newblock Feature construction with version spaces for biochemical
  applications.
\newblock In {\em Proceedings of the Eighteenth International Conference on
  Machine Learning}, ICML '01, pages 258--265, San Francisco, CA, USA, 2001.
  Morgan Kaufmann Publishers Inc.

\bibitem{ExplicitEmbedding_Sidere}
N.~Sid\`ere, Pierre H\'eroux, and J.-Y. Ramel.
\newblock Vector representation of graphs: Application to the classification of
  symbols and letters.
\newblock In {\em Proceedings of the Internation Conference on Document
  Analysis and Recognition}, pages 681--685, 2009.

\bibitem{SpectralEmbedding_Ren}
Peng Ren, R.C. Wilson, and E.R. Hancock.
\newblock Graph characterization via ihara coefficients.
\newblock {\em Neural Networks, IEEE Transactions on}, 22(2):233--245, Feb
  2011.

\bibitem{SpectralEmbedding_Shi}
Jianbo Shi and Jitendra Malik.
\newblock Normalized cuts and image segmentation.
\newblock {\em IEEE Trans. Pattern Anal. Mach. Intell.}, 22(8):888--905, August
  2000.

\bibitem{MCS_Bunke}
Horst Bunke and Kim Shearer.
\newblock A graph distance metric based on the maximal common subgraph.
\newblock {\em Pattern Recogn. Lett.}, 19(3-4):255--259, March 1998.

\bibitem{MCS_Fernandez}
Mirtha-Lina Fernández and Gabriel Valiente.
\newblock A graph distance metric combining maximum common subgraph and minimum
  common supergraph.
\newblock {\em Pattern Recognition Letters}, 22(6–7):753 -- 758, 2001.

\bibitem{Matching_Umeyama}
S.~Umeyama.
\newblock An eigendecomposition approach to weighted graph matching problems.
\newblock {\em Pattern Analysis and Machine Intelligence, IEEE Transactions
  on}, 10(5):695--703, Sep 1988.

\bibitem{Matching_Gold}
S.~Gold and A.~Rangarajan.
\newblock A graduated assignment algorithm for graph matching.
\newblock {\em Pattern Analysis and Machine Intelligence, IEEE Transactions
  on}, 18(4):377--388, Apr 1996.

\bibitem{Matching_Wyk}
Barend~J. van Wyk and Michael~A. van Wyk.
\newblock A pocs-based graph matching algorithm.
\newblock {\em IEEE Trans. Pattern Anal. Mach. Intell.}, 26(11):1526--1530,
  November 2004.

\bibitem{Bunke:1983:IGM:2305869.2306079}
H~Bunke and G~Allermann.
\newblock Inexact graph matching for structural pattern recognition.
\newblock {\em Pattern Recogn. Lett.}, 1(4):245--253, May 1983.

\bibitem{DBLP:journals/cviu/RaveauxAHT11}
Romain Raveaux, S{\'e}bastien Adam, Pierre H{\'e}roux, and {\'E}ric Trupin.
\newblock Learning graph prototypes for shape recognition.
\newblock {\em Computer Vision and Image Understanding}, 115(7):905--918, 2011.

\bibitem{Riesen:2007:GEV:1769371.1769413}
Kaspar Riesen, Michel Neuhaus, and Horst Bunke.
\newblock Graph embedding in vector spaces by means of prototype selection.
\newblock In {\em Proceedings of the 6th IAPR-TC-15 International Conference on
  Graph-based Representations in Pattern Recognition}, GbRPR'07, pages
  383--393, Berlin, Heidelberg, 2007. Springer-Verlag.

\bibitem{Bunke2012811}
Horst Bunke and Kaspar Riesen.
\newblock Towards the unification of structural and statistical pattern
  recognition.
\newblock {\em Pattern Recognition Letters}, 33(7):811--825, 2012.

\bibitem{DBLP:conf/gbrpr/RiesenNB07}
Kaspar Riesen, Michel Neuhaus, and Horst Bunke.
\newblock Bipartite graph matching for computing the edit distance of graphs.
\newblock In {\em Graph-Based Representations in Pattern Recognition, 6th
  IAPR-TC-15 International Workshop, GbRPR 2007, Alicante, Spain, June 11-13,
  2007, Proceedings}, pages 1--12, 2007.

\bibitem{DBLP:journals/ijprai/FankhauserRBD12}
Stefan Fankhauser, Kaspar Riesen, Horst Bunke, and Peter~J. Dickinson.
\newblock Suboptimal graph isomorphism using bipartite matching.
\newblock {\em IJPRAI}, 26(6), 2012.

\bibitem{DBLP:journals/ivc/RiesenB09}
Kaspar Riesen and Horst Bunke.
\newblock Approximate graph edit distance computation by means of bipartite
  graph matching.
\newblock {\em Image Vision Comput.}, 27(7):950--959, 2009.

\bibitem{DBLP:conf/sspr/NeuhausRB06}
Michel Neuhaus, Kaspar Riesen, and Horst Bunke.
\newblock Fast suboptimal algorithms for the computation of graph edit
  distance.
\newblock In {\em SSPR/SPR}, pages 163--172, 2006.

\bibitem{10.1109/34.862201}
Richard Myers, Richard~C. Wilson, and Edwin~R. Hancock.
\newblock Bayesian graph edit distance.
\newblock {\em IEEE Transactions on Pattern Analysis and Machine Intelligence},
  22(6):628--635, 2000.

\bibitem{DBLP:journals/prl/RaveauxBO10}
Romain Raveaux, Jean-Christophe Burie, and Jean-Marc Ogier.
\newblock A graph matching method and a graph matching distance based on
  subgraph assignments.
\newblock {\em Pattern Recognition Letters}, 31(5):394--406, 2010.

\bibitem{DBLP:conf/mlg/RiesenFB07}
Kaspar Riesen, Stefan Fankhauser, and Horst Bunke.
\newblock Speeding up graph edit distance computation with a bipartite
  heuristic.
\newblock In Paolo Frasconi, Kristian Kersting, and Koji Tsuda, editors, {\em
  Mining and Learning with Graphs, MLG 2007, Firence, Italy, August 1-3, 2007,
  Proceedings}, 2007.

\bibitem{DBLP:conf/gbrpr/FankhauserRB11}
Stefan Fankhauser, Kaspar Riesen, and Horst Bunke.
\newblock Speeding up graph edit distance computation through fast bipartite
  matching.
\newblock In {\em Graph-Based Representations in Pattern Recognition - 8th
  IAPR-TC-15 International Workshop, GbRPR 2011, M{\"u}nster, Germany, May
  18-20, 2011. Proceedings}, pages 102--111, 2011.

\bibitem{citeulike:809181}
D.~Conte, P.~Foggia, C.~Sansone, and M.~Vento.
\newblock {Thirty Years Of Graph Matching In Pattern Recognition}.
\newblock {\em International Journal of Pattern Recognition and Artificial
  Intelligence}, 2004.

\bibitem{livi13}
Lorenzo Livi and Antonello Rizzi.
\newblock The graph matching problem.
\newblock {\em Pattern Analysis and Applications}, 16(3):253--283, 2013.

\bibitem{Gao:2010:SGE:1714377.1714380}
Xinbo Gao, Bing Xiao, Dacheng Tao, and Xuelong Li.
\newblock A survey of graph edit distance.
\newblock {\em Pattern Anal. Appl.}, 13(1):113--129, January 2010.

\bibitem{Zeng09comparingstars:}
Zhiping Zeng, Anthony K.~H. Tung, Jianyong Wang, Jianhua Feng, and Lizhu Zhou.
\newblock Comparing stars: On approximating graph edit distance.
\newblock In {\em Proceedings of the {VLDB} Endowment}, volume~2, pages 25--36,
  2009.

\bibitem{4082128}
P.E. Hart, N.J. Nilsson, and B.~Raphael.
\newblock A formal basis for the heuristic determination of minimum cost paths.
\newblock {\em Systems Science and Cybernetics, IEEE Transactions on},
  4(2):100--107, 1968.

\bibitem{Fischer2015331}
Andreas Fischer, Ching~Y. Suen, Volkmar Frinken, Kaspar Riesen, and Horst
  Bunke.
\newblock Approximation of graph edit distance based on hausdorff matching.
\newblock {\em Pattern Recognition}, 48(2):331 -- 343, 2015.

\bibitem{Almohamad:1993:LPA:628301.628477}
H.~A. Almohamad and S.~O. Duffuaa.
\newblock A linear programming approach for the weighted graph matching
  problem.
\newblock {\em IEEE Trans. Pattern Anal. Mach. Intell.}, 15(5):522--525, May
  1993.

\bibitem{Justice:2006:BLP:1155317.1155424}
Derek Justice and Alfred Hero.
\newblock A binary linear programming formulation of the graph edit distance.
\newblock {\em IEEE Trans. Pattern Anal. Mach. Intell.}, 28(8):1200--1214,
  August 2006.

\bibitem{DBLP:dblp_conf/sspr/RiesenFB14}
Kaspar Riesen, Andreas Fischer, and Horst Bunke.
\newblock Improving approximate graph edit distance using genetic algorithms.
\newblock In {\em S+SSPR}, pages 63--72, 2014.

\bibitem{Wilson:1997:SMD:262631.262639}
Richard~C. Wilson and Edwin~R. Hancock.
\newblock Structural matching by discrete relaxation.
\newblock {\em IEEE Trans. Pattern Anal. Mach. Intell.}, 19(6):634--648, June
  1997.

\bibitem{garey79:_comput_and_intrac}
Michael~R. Garey and David~S. Johnson.
\newblock {\em Computers and Intractability: A Guide to the Theory of
  NP-Completeness}.
\newblock W. H. Freeman \& Co., New York, NY, USA, 1979.

\bibitem{Wolsey1998}
L.~A. Wolsey.
\newblock {\em {Integer programming}}.
\newblock Wiley-Interscience, New York, NY, USA, 1998.

\bibitem{saigal95:_linear_progr}
R.~Saigal.
\newblock {\em Linear Programming: A Modern Integrated Analysis}.
\newblock Kluwer Academin Publishers, 1995.

\bibitem{fischer2013fast}
Andreas Fischer, Ching~Y Suen, Volkmar Frinken, Kaspar Riesen, and Horst Bunke.
\newblock A fast matching algorithm for graph-based handwriting recognition.
\newblock In {\em Graph-Based Representations in Pattern Recognition}, pages
  194--203. Springer, 2013.

\bibitem{iamdb}
Kaspar Riesen and Horst Bunke.
\newblock Iam graph database repository for graph based pattern recognition and
  machine learning.
\newblock In Niels Vitoria~Lobo, Takis Kasparis, Fabio Roli, JamesT. Kwok,
  Michael Georgiopoulos, GeorgiosC. Anagnostopoulos, and Marco Loog, editors,
  {\em Structural, Syntactic, and Statistical Pattern Recognition}, volume 5342
  of {\em Lecture Notes in Computer Science}, pages 287--297. Springer Berlin
  Heidelberg, 2008.

\bibitem{LeBodic20124214}
Pierre~Le Bodic, Pierre Héroux, Sébastien Adam, and Yves Lecourtier.
\newblock An integer linear program for substitution-tolerant subgraph
  isomorphism and its use for symbol spotting in technical drawings.
\newblock {\em Pattern Recognition}, 45(12):4214 -- 4224, 2012.

\bibitem{GRECdb}
Philippe Dosch and Ernest Valveny.
\newblock Report on the second symbol recognition contest.
\newblock In Wenyin Liu and Josep Lladós, editors, {\em Graphics Recognition.
  Ten Years Review and Future Perspectives}, volume 3926 of {\em Lecture Notes
  in Computer Science}, pages 381--397. Springer Berlin Heidelberg, 2006.

\bibitem{Riesen:2010:GCC:1855255}
Kaspar Riesen and Horst Bunke.
\newblock {\em Graph Classification and Clustering Based on Vector Space
  Embedding}.
\newblock World Scientific Publishing Co., Inc., River Edge, NJ, USA, 2010.

\bibitem{doi:10.1021/jm040835a}
Jeroen Kazius, Ross McGuire, and Roberta Bursi.
\newblock Derivation and validation of toxicophores for mutagenicity
  prediction.
\newblock {\em Journal of Medicinal Chemistry}, 48(1):312--320, 2005.

\bibitem{4202588}
N.A. Fox, R.~Gross, J.F. Cohn, and R.B. Reilly.
\newblock Robust biometric person identification using automatic classifier
  fusion of speech, mouth, and face experts.
\newblock {\em Multimedia, IEEE Transactions on}, 9(4):701--714, June 2007.

\bibitem{Schomburg:2004nu}
I.~Schomburg, A.~Chang, C.~Ebeling, M.~Gremse, C.~Heldt, G.~Huhn, and
  D.~Schomburg.
\newblock Brenda, the enzyme database: updates and major new developments.
\newblock {\em Nucleic Acids Research}, 32(suppl 1):D431--D433, 2004.

\bibitem{Erdos60onthe}
P.~Erdős and A~Rényi.
\newblock On the evolution of random graphs.
\newblock In {\em Publication of the mathematical institute of the Hungarian
  academy of sciences}, pages 17--61, 1960.

\bibitem{DBLP:journals/scheduling/BaptisteCGT10}
P.~Baptiste, Federico~Della Croce, Andrea Grosso, and Vincent T'Kindt.
\newblock Sequencing a single machine with due dates and deadlines: an
  ilp-based approach to solve very large instances.
\newblock {\em J. Scheduling}, 13(1):39--47, 2010.

\end{thebibliography}

\end{document}